\documentclass[pdflatex,sn-mathphys-num]{sn-jnl}

\usepackage{graphicx}%
\usepackage{multirow}%
\usepackage{amsmath,amssymb,amsfonts}%
\usepackage{amsthm}%
\usepackage{mathrsfs}%
\usepackage[title]{appendix}%
\usepackage{xcolor}%
\usepackage{textcomp}%
\usepackage{manyfoot}%
\usepackage{booktabs}%
\usepackage{listings}%

\usepackage{bbm}
\usepackage{braket}
\usepackage{enumitem}
\usepackage[capitalize]{cleveref}

\theoremstyle{thmstyleone}%
\newtheorem{thm}{Theorem}
\newtheorem{prop}[thm]{Proposition}

\newtheorem{lemma}[thm]{Lemma}

\theoremstyle{thmstyletwo}%
\newtheorem{rem}[thm]{Remark}%

\theoremstyle{thmstylethree}%
\newtheorem{defn}[thm]{Definition}%

\newcommand{\be}{\begin{equation}}
\newcommand{\ee}{\end{equation}}
\newcommand{\mc}[1]{\mathcal{#1}}
\newcommand{\mf}[1]{\mathfrak{#1}}

\def\A{{\cal A}} 
\def\a{{A}} 
\def\M{{\mathscr M}} 
\def\Reg{{\mathscr R}} 
\def\Cou{{\mathscr K}} 
\def\Pro{{\mathscr L}} 
\def\J{{\mathscr J}} 
\renewcommand{\S}{\mc{T}} 
\renewcommand{\P}{\mc{P}} 
\newcommand{\Out}{\mf{X}} 
\newcommand{\out}{s} 
\newcommand{\Instr}{I} 

\def\SmoothC{{\mathcal D}} 
\def\Smooth{{\mathcal C}^\infty} 

\newcommand{\spl}{\xi} 

\def\N{\mathbbm{N}} 
\newcommand{\R}{\mathbbm{R}} 
\renewcommand{\H}{\mc{H}} 
\newcommand{\C}{\mathbbm{C}} 

\newcommand{\mat}[1]{\left( \begin{matrix}#1\end{matrix}\right)}

\newcommand{\non}{\nonumber}

\newcommand{\defby}{:=}

\DeclareMathOperator{\supp}{supp}

\newcommand{\id}{\mathbbm{1}}

\newcommand{\pl}{(} 
\newcommand{\pr}{)}

\renewcommand{\orcidlogo}{%
  \includegraphics[width=10pt]{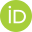}%
}
\renewcommand{\orcid}[1]{\href{https://orcid.org/#1}{\orcidlogo}}

\begin{document}

\title[Quantum Field Measurements in the Fewster-Verch Framework]{Quantum Field Measurements in the Fewster-Verch Framework}

\author*[1]{\fnm{Jan} \sur{Mandrysch} \orcid{0000-0003-4827-1806}}\email{\normalsize jan.mandrysch@oeaw.ac.at}

\author[1]{\fnm{Miguel} \sur{Navascu\'{e}s} \orcid{0000-0003-0717-3927}}

\affil[1]{\normalsize \orgdiv{Institute for Quantum Optics and Quantum Information (IQOQI) Vienna},\\ \orgname{Austrian Academy of Sciences}, \orgaddress{\street{Boltzmanngasse 3}, \postcode{1090} \city{Wien}, \country{Austria}}}


\abstract{The Fewster-Verch (FV) framework provides a local and covariant approach for defining measurements in quantum field theory (QFT). Within this framework, a probe QFT represents the measurement device, which, after interacting with the target QFT, undergoes an arbitrary local measurement. Remarkably, the FV framework is free from Sorkin-like causal paradoxes and robust enough to enable quantum state tomography. However, two open issues remain.

First, it is unclear if the FV framework allows conducting arbitrary local measurements. Second, if the probe field is interpreted as physical and the FV framework as fundamental, then one must demand the probe measurement to be itself implementable within the framework. That would involve a new probe, which should also be subject to an FV measurement, and so on. It is unknown if there exist non-trivial FV measurements for which such an ``FV-Heisenberg cut" can be moved arbitrarily far away. In this work, we advance the first problem by proving that Gaussian-modulated measurements of locally smeared fields fit within the FV framework. We solve the second problem by showing that any such measurement admits a movable FV-Heisenberg cut.

As a technical byproduct, we establish that state transformations induced by finite-rank perturbations of the classical phase space underlying a linear scalar field preserve the Hadamard property.}



\maketitle

\newpage 

\section{Introduction}
Quantum Field Theory (QFT) has proven very successful in modeling particle scattering experiments at high energies. However, as noted by Sorkin \cite{Sor93}, the use of seemingly innocuous quantum operations within the QFT formalism would allow three or more separate parties to violate Einstein's causality. The accepted resolution of Sorkin's paradox is to acknowledge that the set of quantum operations which can be conducted within a finite region of space-time is smaller than previously envisaged. This raises the question of how to model local operations within QFT.

This problem has been greatly advanced in recent years. In \cite{Jub22}, Jubb studies which quantum channels (or update rules) do not lead to Sorkin-type violations of causality in QFT, a property he calls strong causality. Among other things, he concludes that weak ``Gaussian'' measurements of locally smeared fields are strongly causal. Oeckl, who calls this property causal transparency, reaches the same conclusion in \cite{Oec24}.

Independently, Fewster and Verch proposed in \cite{FV20,FV23} a general framework to model measurements within QFT. In a nutshell, the FV framework consists in making the `target' (or `system') field theory interact in a compact region with another `probe' field theory. Measurement outcomes are obtained by measuring the probe in an arbitrary bounded region outside the causal past of the interaction region. The probe is then discarded, i.e., one is not allowed to measure it a second time. Remarkably, all quantum operations achievable within the FV framework are strongly causal \cite{BFR21}. Moreover, they suffice to carry out full tomography when the target QFT is a scalar field \cite{FJR23}. However, as we argue later on, this is different from modeling actual measurements. In this regard, the FV framework presents at least two unresolved issues:

\begin{itemize}
    \item First, given the specification of a quantum measurement, or even of the corresponding quantum instrument, it is not straightforward to determine whether the FV framework allows implementing one or the other---even approximately. Doing so involves specifying a probe QFT, its initial state, the form of the interaction, and the measurement on the probe, or proving that no such elements exist.
    \item Second, if we regard the probe QFT as a physical entity and not as part of a mere mathematical construction to obtain a set of strongly causal operations, then it is unclear how come one is allowed to measure the probe arbitrarily; a standing assumption of the framework. Consistency therefore demands that the probe measurement itself be modeled within the FV framework. That would involve introducing another probe, whose measurement should be realizable within the FV framework, and so on. In other words, any physical measurement should be expressible as an arbitrarily long chain of consecutive interactions with independent probes, in such a way that we predict the same measurement statistics irrespective of where in the chain we invoke Born's rule.

    In non-relativistic quantum mechanics, the specific probe within a measurement chain where one applies Born's rule is known as ``Heisenberg cut". That the position of the cut can be chosen at will was already observed by von Neumann \cite{Von96}. What we have argued above is that, if the FV framework is fundamental and its probe field is meant to be physical, then feasible QFT measurements should have a similarly movable ``FV-Heisenberg cut". This leads us to ask: Are there non-trivial FV-realizable measurements with a movable FV-Heisenberg cut?
\end{itemize}

In this paper, we address the first problem by proving that Gaussian-modulated measurements of a locally smeared quantum field can be modeled within the FV framework. As it turns out, the probe measurement that induces a Gaussian-modulated field measurement can itself be chosen a Gaussian-modulated field measurement. This solves the second problem: Gaussian field measurements are consistent with the FV framework, even if we accept the physicality of the probe field. As a byproduct, we establish that projective field measurements can also be modeled within the FV framework.

The structure of this paper is as follows. In \Cref{sec:alg} we specify the linear QFT setting to which our results apply. In \Cref{sec:quantumops} we introduce how to model general as well as Gaussian-modulated measurements in QFT, distinguishing between positive-operator-valued measures (POVMs) and quantum instruments, and illustrate how this differs from state tomography. In \Cref{sec:fv} we summarize the FV framework and characterize the POVMs and instruments that are realizable within the framework in order to state our main results (\Cref{thm:asymptFV}). Finally, in \Cref{sec:impl} we prove our results by presenting a concrete scheme to induce Gaussian-modulated measurements, projective measurements, and more---even iteratively. We conclude in \Cref{sec:conclusion}. In Appendix~\ref{app:aqft} we review the construction of the Weyl CCR algebra of a linear scalar field. In Appendix~\ref{app:finrank}, we establish that state transformations induced by finite-rank perturbations of the classical phase space underlying a linear scalar field preserve the Hadamard property. In Appendix~\ref{app:gauss} we collect some proofs on (dephased) Gaussian instruments.

\section{Technical setting for quantum field theory}\label{sec:alg}

We will be dealing with the standard framework of AQFT: namely, starting from a fixed globally hyperbolic spacetime manifold $\M$, we associate to any region $\Reg\subset \M$ a unital *-algebra $\A(\Reg)$. These local algebras, which define our QFT, satisfy the usual axioms: isotony, microcausality, time-slice and Haag's property; confer \cite[Sec.~4]{FV23} and references therein for details. States of this QFT are linear, positive functionals $\omega:\A\to\C$ on the algebra of observables $\A \defby \A(\M)$ with $\omega(1)=1$. An element $A \in \A$ is said to be localizable in region $\Reg\subset \M$ if $\a\in\A(\Reg)$. To express limits, we assume $\A$ and all $\A(\Reg)$ are equipped and completed with respect to a suitable topology. Specifically, we work with von-Neumann algebras; for some fixed Hilbert space $\H$ these are *-subalgebras of $B(\H)$ closed under the strong-operator topology\footnote{..., or closed under the weak-operator-topology, a more common but equivalent definition of a von-Neumman algebra.}. Limits, infinite sums, and integrals pertaining to $\A$ will be taken with respect to this topology, with integrals understood as Bochner integrals in this topology. Note that in this setting, any density matrix, i.e., $\rho \in B(\H)$ such that $\rho \geq 0$ and $\operatorname{tr} \rho = 1$, gives rise to a state $\omega_\rho$ on $\A$ through $\omega_\rho(\a ) \defby \operatorname{tr} (\rho \a )$.

Our main study will be conducted in the setting of linear scalar QFT. In this setting the \emph{smeared field} $\Phi(f)$ corresponds to the quantization of
\begin{equation}
    \phi(f) = \int_\M \phi(x) f(x) \mathrm{vol}(x) =: \braket{\phi,f},
\end{equation}
a classical point field $\phi$ satisfying a classical linear differential equation $T\phi(x)= 0$ and smeared over spacetime through a test function $f$ using $\mathrm{vol}$, the volume form\footnote{Its dependence on the spacetime metric is suppressed from our notation for $\M$ since it is considered to be fixed throughout the article.} on $\M$. We assume that $T$ is symmetric with respect to $\braket{,}$ and normally hyperbolic, which implies it has well-defined retarded/advanced Green operators $E^\pm$ and associated causal propagator $E= E^--E^+$ \cite{BGP07}. This holds, e.g., for the Klein-Gordon operator $T = \square_\M - m^2$, $m > 0$. We denote the space of real-valued smooth compactly supported functions on $\M$ by $\SmoothC(\M)$ and by $\SmoothC(\Reg)$ if restricting to functions having support in $\Reg \subset \M$. With $\phi$ having spatially compact support, it is natural to identify $Tf \sim 0$ for all $f \in \SmoothC(\M)$ since $\phi(Tf) = \braket{T\phi,f} = 0$ by partial integration. We denote the resulting quotient space by $\mc{C}_T \defby \SmoothC(\M)/T\SmoothC(\M)$ and say that $f \in \SmoothC(\M)$ is \emph{localizable in }$\Reg \subset \M$ iff there exists $g \in \SmoothC(\Reg)$ such that $f \sim g$. In this case we denote $[f] \in \mc{C}_T(\Reg) \subset \mc{C}_T$. The \emph{classical timeslice property} \cite[Thm.~3.3.1]{BDFY15} entails that any $f\in \SmoothC(\M)$ is localizable in any globally hyperbolic neighbourhood $N$ of any Cauchy surface of $M$. We define the \emph{classical phase space} as the symplectic space $(\mc{C}_T,E)$, where $E$ denotes also the symplectic form $E([f],[g]) \defby \braket{f, Eg}$, $f,g \in \SmoothC(\M)$ , associated with the causal propagator. This concludes the classical part of the construction.

Quantization may then be performed in a fully algebraic setting which we allude to in Appendix~\ref{app:aqft}. For the main text, however, we are content with a fixed Hilbert space representation and make these assumptions: The \emph{smeared field} $\Phi$ is a map from $\mc{C}_T$ to self-adjoint operators on a Hilbert space $\H$ such that
\begin{align}
  \label{eq:cont}&\textbf{[Continuity]} \, && e^{it\Phi(f)} \text{ is strongly operator continuous in }t,\\
  \label{eq:weyl}&\textbf{[Weyl relation]} \, &&e^{i\Phi(f)} e^{i\Phi(g)} = e^{-\frac{i}{2} \braket{f,Eg}} e^{i\Phi(f+g)},
\end{align}
for all $f,g \in \SmoothC(\M)$, where we agree to $\Phi(f)\defby \Phi([f])$ by abuse of notation. Note here that the Weyl relation is a strong and covariant form of the canonical commutation relations. Finally, we obtain a net of von-Neumann algebras $\{ \A(\Reg)\}_\Reg \subset B(\H)$ labeled by open causally convex regions $\Reg \subset \M$ by defining $\A(\Reg)$ to be the algebra generated by the \emph{Weyl operators} $\{ e^{i\Phi(f)}: \, f \in \SmoothC(\Reg)\}$, equipped and completed with respect to the strong operator topology of $B(\H)$. These local algebras are known to satisfy the usual axioms as listed above. We will call this a \emph{linear scalar QFT} (induced by $T$) and denote it by $\A_T$.

We conclude this section by specifying suitable states and state transformations on $\A_T$. For QFTs (at least for linear ones) there is a broad class of physically reasonable states given by the quasi-free Hadamard states. There are many such states on any globally hyperbolic spacetime \cite{Ver94} and for stationary spacetimes this includes ground and thermal states \cite{San13}. There is also a broad class of state transformations: Any symplectic transformation $F$ of $(\mc{C}_T,E)$ induces an automorphism $\alpha_F$ of $\A_T$ by imposing $\alpha_F(e^{i\Phi(v)}) \defby e^{i\Phi(Fv)}$, $v \in \mc{C}_T$ and a new \emph{$F$-transformed state} $\omega_F$ via pullback, i.e.,
\begin{equation} \label{eq:transformedstate}
    \omega_F(\bullet) \defby \omega(\alpha_F(\bullet)).
\end{equation}
To deem a transformation of the form
\begin{equation} \label{eq:statetrafo}
    \omega \mapsto \omega_F
\end{equation}
as `physical', e.g., arising as a state-update for a localized measurement, one would at least require that it preserves the class of physically reasonable states. 

It is well known (confer, e.g., \cite[Prop.~29.3-2]{HR15}) that the class of quasi-free states is preserved by \emph{all} transformations which have the form \eqref{eq:statetrafo}. In Appendix~\ref{app:finrank} we prove that these preserve also the Hadamard class if $F$ transforms only a finite-dimensional subspace of $\mc{C}_T$. Although possibly known to experts, this result, to the best of our knowledge, has not been addressed in the existing literature and might be of independent interest. 

Since our main study does not require to work with Hadamard states, we will only define quasi-free states in more detail: A state $\omega$ on $\A_T$ is referred to as \emph{quasi-free} (sometimes also as ``Gaussian'') iff there exists a positive bilinear symmetric form $\Gamma:\mc{C}_T^{\times 2}\to \R$  such that 
\begin{equation}\label{eq:qfree}
\omega(e^{i\Phi(v)})=e^{-\frac{\Gamma(v,v)}{2}}, \quad v \in \mc{C}_T.
\end{equation}
Conversely, \eqref{eq:qfree} defines a state on $\A_T$ iff 
\begin{equation} \label{eq:uncertaintyrel}
    \Gamma(v,v) \Gamma(w,w) \geq \tfrac{1}{4} E(v,w)^2, \quad v,w \in \mc{C}_T.
\end{equation}
We refer to such $\Gamma$ as a \emph{covariance}. Given any symplectic transformation $F$ of $(\mc{C}_T,E)$, it is straightforward to check that the $F$-transformed state $\omega_F$ is also quasi-free with covariance
\begin{equation}
    \Gamma_F(v,w) \defby \Gamma(Fv,Fw), \quad v,w \in \mc{C}_T.
\end{equation}
A quasi-free state is called \emph{pure} iff \eqref{eq:uncertaintyrel} can be saturated, i.e., for any given nonzero $v \in \mc{C}_T$, 
\begin{equation}\label{eq:optuncertainty}
    \underset{0\neq w \in \mc{C}_T}{\sup}\, \frac{E(v,w)^2}{4 \Gamma(v,v)\Gamma(w,w)} = 1.
\end{equation}
Note that also \emph{pure} quasi-free Hadamard states exist in arbitrary globally hyperbolic spacetimes. In stationary spacetimes, these exist since, at least for a free scalar field with a positive stationary potential, i.e., where $T= \square_\M + V$ with stationary $0 < V \in C^\infty(\M)$, there is a unique (pure) ground state \cite{San13}. More generally, existence follows by a deformation argument\footnote{We thank Chris~Fewster for the helpful suggestion.}: For any globally hyperbolic spacetime $\M$, there is a deformation in the form of a certain time-orientation preserving diffeomorphism such that the deformed spacetime, say $\M'$, is static in the causal past of some Cauchy surface \cite{FNW81}. A quasi-free Hadamard state on $\M$ may then be constructed by taking suitable extensions and restrictions of a pure quasi-free Hadamard state on the static patch of $\M'$. That this state is also pure follows using the defining relation \eqref{eq:optuncertainty} and the classical timeslice property.

\newpage

\section{Quantum field measurements and instruments}\label{sec:quantumops}

We wish to model local operations within QFT, specifically measurements and associated instruments. In particular, we will use the technical setting from the previous section and fix some QFT with algebra of observables $\A$. Later on we illustrate these notions defining Gaussian operations in a linear scalar QFT.

Generally, a (quantum) measurement comes with a set of (classical) \emph{measurement outcomes} $\Out$ and associates with it a (classical) probability distribution which depends on the state of the (quantum) system---in our case a QFT---before the measurement. If we add a description of the state of the system after the measurement, we speak of a (quantum) instrument. To avoid technicalities, we treat the two paradigmatic cases of a \emph{finite-} and a \emph{continuous-}outcome measurement and refer to \cite{BLPY16} for generalizations. Moreover, as it is convenient, we define our instruments in the Heisenberg picture acting directly on $\A$---confer e.g. \cite{Wer01,Key02}---and provide the associated state changes in form of \emph{state update rules}.

To begin with, let $\Out$ be \emph{finite}. In this case, an \emph{instrument} is defined by a family of linear completely positive maps $\Instr \defby \{ \Instr_j \}_{j\in \Out}$, where $\Instr_j :\A\to\A$ is such that
\begin{equation}\label{eq:sum1}
    \bar{\Instr}(\a ) \defby \sum_j  \Instr_j (\a ) \quad\text{satisfies}\quad \bar{\Instr}(1) = 1.
\end{equation}
The map $\bar{\Instr}: \A \to \A$ is called the \emph{measurement channel}. Intuitively, if our QFT is initially in state $\omega$, the instrument $\Instr$ describes a measurement that returns an outcome $j$ with probability
\begin{equation}\label{prob:1}
p(j |\omega)=\omega(\Instr_j (1)),
\end{equation}
in which case the new state of the QFT is updated to
\begin{equation}
\bar{\omega}_j  \defby \frac{1}{\omega(\Instr_j (1))}\omega\circ\Instr_j .
\end{equation}
If we ignore or discard the measurement outcome $j$, then the new QFT state will be
\begin{equation}\label{eq:sum2}
\bar{\omega} \defby \sum_j p(j |\omega)\bar{\omega}_j =\omega\circ\bar{\Instr}.
\end{equation}

If we are just interested in computing the probabilities of the different outcomes, it suffices to work with the algebra elements $M \defby \{ M_j \defby I_j(1) \}_j$ so that $p(j |\omega)=\omega(M_j )$. These elements satisfy
\begin{equation} \label{eq:sum3}
  \sum_j  M_j =1 \quad \text{and}\quad \forall j : \, M_j \geq 0.
\end{equation}
Any such family $M$ is called a \emph{positive operator-valued measure} (POVM).

The formalism above can also be used to model \emph{continuous}-outcome measurements. In this case we suppose the set of outcomes $\Out \subset \R$ to be a Borel subset of the real line, e.g., an interval or the full real line, we label the outcomes with indices $\out$ instead of $j$ for distinction\footnote{In a few cases we also use $\alpha$ as index if we simultaneously treat finite and continuous outcomes.}and we replace $\sum_{j\in\Out} $ by $\int_{\Out} \bullet \,d\out $ whenever it appears---in particular within \eqref{eq:sum1}, \eqref{eq:sum2} and \eqref{eq:sum3}. Further, we require that $\Instr_\out (\a)$ is integrable in $\out $ for all $\a \in \A$. Here, "integrable" refers to Bochner integrability with respect to the strong-operator topology (confer Appendix~\ref{app:gauss} for a definition). Note that this assumes more than strictly necessary to keep the statistical interpretation from above, Eqs.~\eqref{prob:1}--\eqref{eq:sum2}, which relies only on the existence of limits within states; for example, the weak-operator topology would be sufficient in this regard. However, our choice of topology implies the existence of the measurement channel, $\overline{\Instr}(\a) \defby \int_{\Out} \Instr_\out (\a) d\out $ as a completely positive map within $\A$ and we will establish this property for all given examples.

\subsection{Gaussian-modulated measurements}
Now, we are ready to discuss concrete examples of instruments and POVMs. In order to do this, we take $\A$ to be a linear scalar QFT as specified in the preceding section (dropping the index $T$) and denote the associated quantized field by $\Phi$. We remark that expressions of the form $\mu(\Phi(f))$ for any bounded Borel function $\mu:\R \to \mathbbm{C}$, by means of the spectral theorem, define elements of $\A \subset B(\H)$ in terms of (Borel) functional calculus. From this perspective, most statements made in the remainder of this section are direct consequences of the properties of the functions $\mu$, and we defer proofs to Appendix~\ref{app:gauss}. 

To start with, let $f\in \SmoothC(\M)$ and $\epsilon>0$, and consider the following continuous-outcome POVM on $\A$,
\begin{equation}
M^{f,\epsilon} \defby \{ M_\out^{f,\epsilon}\}_{\out \in\R}, \quad M_\out^{f,\epsilon} \defby \frac{1}{\sqrt{2\pi}\epsilon} e^{-\frac{(\Phi(f)-\out )^2}{2\epsilon^2}}.
\label{gaussian_POVM}
\end{equation}
We will call this POVM a \emph{Gaussian-modulated measurement of the field $\Phi(f)$}, or Gaussian measurement, for short. It is a special instance of a weak measurement \cite{weak_meas}. Moreover, it approximates a projective measurement of $\Phi(f)$ in the limit $\epsilon \to 0$, i.e., for any Borel set $B \subset \R$ and with $\Pi_B(A)$ denoting the spectral projection of a self-adjoint operator $A$ associated with the spectral subset $B$,
\begin{equation} \label{eq:gaussPOVMlimit}
  \Pi_B(\Phi(f))=\underset{\epsilon\to 0}{\lim} \, \int_B M_\out^{f,\epsilon} d\out .
\end{equation}

Note that there exist many different instruments implementing the Gaussian POVM~\eqref{gaussian_POVM}. One of them is:
\begin{equation}
\Instr^{f,\epsilon}_\out (\bullet) \defby \frac{1}{\sqrt{2\pi}\epsilon} e^{-\frac{(\Phi(f)-\out )^2}{4\epsilon^2}} \bullet e^{-\frac{(\Phi(f)-\out )^2}{4\epsilon^2}}.
\label{gaussian_instrument_basic}
\end{equation} 
This instrument appears in the literature of quantum optics and quantum foundations \cite{GRW86,GI02}. Its causality properties have been recently studied in QFT, see \cite{Jub22} and \cite{Oec24}.

Now, for any $\delta>0$, define the \emph{dephasing map}
\begin{equation}\label{eq:dephasing}
D^{f,\delta}(\bullet) \defby \frac{1}{\sqrt{2\pi}\delta}\int e^{-\frac{\nu^2}{2\delta^2}}e^{-i\nu\Phi(f)}\bullet e^{i\nu\Phi(f)}  d\nu.
\end{equation}
This map is completely positive and unit preserving. Intuitively, it introduces decoherence between the generalized eigenstates of $\Phi(f)$.

It is easy to verify that the instrument
\begin{equation}
\Instr^{f,\epsilon,\delta} \defby \{ \Instr^{f,\epsilon,\delta}_\out  \}_{\out \in\R}, \quad \Instr^{f,\epsilon,\delta}_\out  \defby \Instr_\out ^{f,\epsilon} \circ D^{f,\delta} =D^{f,\delta}\circ \Instr_\out ^{f,\epsilon}
\label{eq:dephasedgauss_instr}
\end{equation}
also induces the POVM $M^{f,\epsilon}$ defined in (\ref{gaussian_POVM}) for any $\delta > 0$. This example will have a prominent appearance in our measurement scheme and illustrates that different instruments might induce the same POVM.

The instrument $\Instr^{f,\epsilon,\delta}$ is fully characterized by its action on Weyl operators. In this regard, by a straightforward computation it follows that
\begin{equation}
  \Instr^{f,\epsilon,\delta}_\out (e^{i\Phi(h)})=\frac{1}{\sqrt{2\pi}\epsilon} e^{-\frac{\langle f,Eh\rangle^2}{8\underline{\epsilon}(\epsilon, \delta)^2}}e^{i\Phi(h)}e^{-\frac{(\Phi(f)-\out -\frac{\langle f,Eh \rangle}{2})^2}{2\epsilon^2}}, \quad \underline{\epsilon}(\epsilon, \delta)^2 \defby \frac{1}{4\delta^2+\frac{1}{\epsilon^2}}
\label{action_gauss_instr}
\end{equation}
for arbitrary $h \in \SmoothC(\M)$. Integrating over $\out $, we arrive at an expression for the measurement channel of instrument \eqref{eq:dephasedgauss_instr}, namely:
\begin{equation}
\bar{\Instr}^{f,\epsilon,\delta}(e^{i\Phi(h)})=e^{-\frac{\langle f,Eh\rangle^2}{8\underline{\epsilon}(\epsilon, \delta)^2}} e^{i\Phi(h)}.
\end{equation}
Hence, the instruments $\Instr^{f,\epsilon,\delta}$ and $\Instr^{f,\underline{\epsilon}(\epsilon,\delta),0}$, despite giving rise to different POVMs, share the same measurement channel; a measurement channel which is causal, as has been verified in \cite{Jub22} and \cite{Oec24}.

\subsection{State tomography}
We conclude this section with some words on tomography. Any family of POVMs $\{M^l\}_l$ with elements spanning the Hermitian part of $\A$ allows one to completely fix the underlying state $\omega$. That is, for any other state $\omega'$ on $\A$, the condition
\begin{equation}
\forall \out ,l: \quad \omega(M^l_\out )=\omega'(M^l_\out )
\end{equation}
implies that $\omega=\omega'$. In that case, we say that the family $\{M^l\}_l$ is \emph{tomographically complete}. 

Being able to implement a tomographically complete family of POVMs does not necessarily mean being able to implement, even approximately, any possible POVM. For instance, suppose that, for fixed $\epsilon>0$, we can measure all POVMs of the form $\{M^{f,\epsilon} :\, f\in \SmoothC(\M),\,\langle Ef,\alpha Ef\rangle=1\}$, for some positive $\alpha\in L^1(\M)$. Now, take $f\in \SmoothC(\M)$, with $\langle Ef,\alpha Ef\rangle=1$. Suppose for an ensemble of independent preparations of a quasi-free state $\omega$ that we measure $M^{f,\epsilon}$ and, given the result $\out $, compute $e^{i\lambda \out }$, $\lambda \in \R$. Then, the resulting random variable $\beta(f,\lambda)$ has the expectation value
\begin{equation}
\langle \beta(f,\lambda)\rangle =\frac{1}{\sqrt{2\pi}\epsilon}\int \omega\left( e^{-\frac{(\Phi(f)-\out )^2}{2\epsilon^2}}\right)e^{i\lambda \out } d\out =e^{-\frac{\epsilon^2\lambda^2}{2}}\omega\left( e^{i\lambda \Phi(f)}\right).
\end{equation}
Since the linear span of all Weyl operators yields a dense subset of $\A$, it follows that the given family of POVMs is tomographically complete. Now, consider the spectral projectors 
\begin{equation} \label{eq:proj1}
\Pi_0= \Pi_{[-1,1]}(\Phi(f)),\quad \Pi_1=1-\Pi_0,
\end{equation}
which form the POVM $\Pi:=\{\Pi_0,\Pi_1\}$. Then we have that
\begin{equation}
\frac{1}{\pi} \lim_{\Lambda \to \infty} \, \int_{-\Lambda}^{\Lambda} \mbox{sinc}\left(\lambda\right) e^{\frac{\epsilon^2\lambda^2}{2}}\langle \beta(f,\lambda)\rangle d\lambda=\omega\left(\Pi_0\right).
\end{equation}
That is, given the statistics obtained by measuring independent preparations of $\omega$ with $M^{f,\epsilon}$, one can estimate the probability with which one would have observed outcome $0$ (or $1$) had we measured $\Pi$.

This is, however, not the same as measuring the POVM $\Pi$. In fact, by measuring $\Pi$, one can distinguish with certainty if the system was prepared in state $\omega=\omega_\rho$ with $\rho=\ket{\psi_0}\bra{\psi_0}$ or $\rho=\ket{\psi_1}\bra{\psi_1}$ for any two vectors $\ket{\psi_0},\ket{\psi_1}$ with $\Pi_j\ket{\psi_j}=\ket{\psi_j}$. In contrast, by measuring $M^{f, \epsilon}$, one cannot tell with certainty which of the two states was prepared. Depending on the value of $\epsilon$ one would need to measure many independent preparations of $\omega$ with $M^{f, \epsilon}$ to reach a statistically significant conclusion.

\section{The FV framework and its ``Heisenberg cut''}\label{sec:fv}

It is tempting to think that any instrument of the form $\Instr_j (\bullet) \defby \sum_i \a_{i,j }\bullet \a^*_{i,j }$ for some (say finite) family $\{\a_{i,j }\}_{i,j }\subset\A(\Reg)$ with $\sum_{i,j } \a_{i,j} \a^*_{i,j }=1$ can be physically implemented by acting only on region $\Reg$ or, at most, its causal hull. However, as shown by Sorkin \cite{Sor93}, doing so leads to contradictions with Einstein's causality in scenarios with three or more separate experimenters. For instance $\Instr_j(\bullet) \defby \Pi_j \bullet \Pi_j$ with POVM $\Pi_j$ as defined in \eqref{eq:proj1} would be acausal. Since Sorkin-like scenarios occur also more widely---recently an example for classical field theory has been given in \cite{MV23}---the issue is not to be attributed to quantum theory or QFT per sé but rather to an insufficient characterization of local instruments which do not clash with causality.

Recently, Fewster and Verch proposed a set of quantum instruments to model measurements \cite{FV20,FV23} which was subsequently found to be strongly causal \cite{BFR21}, avoiding any Sorkin-like paradoxes. Their main idea was to model the measurement of a `target' QFT observable\footnote{In the original article the term `system' is used in place of `target'.} by making it interact with a `probe' QFT in a region $\Cou$. The probe field is measured using a POVM with elements localized in some processing region $\Pro$ strictly separated from the past of $\Cou$. After discarding (`tracing out') the probe field, the resulting quantum instrument induced in the target theory is shown to be localizable in any causally complete connected region strictly containing $\Cou$. 

More formally, call $\S$ the target QFT and $\P$ the probe QFT, which we assume to be in state $\sigma$. The coupling between system and probe is specified in terms of an automorphism $\Theta$ acting on the uncoupled theory $\S\otimes \mc{P}$ implementing the interaction---referred to as the \emph{scattering morphism}. A key assumption is that the \emph{coupling region} $\Cou \subset \M$ associated with $\Theta$ is compact, both in space and time, which allows the identification of the coupled theory with the uncoupled one, $\S\otimes \mc{P}$, in the so-called ``in'' and ``out'' regions $\Cou^\pm \defby \M \setminus \J^{\mp}(\Cou)$, where $\J^\pm(\Reg)$ denotes the causal future/past of a region $\Reg \subset \M$. The \emph{processing region} $\Pro$ is assumed to be precompact such that $\overline{\Pro}\subset \Cou^+$.

If we now implement a POVM $M \subset \mc{P}(\Pro)$ on the probe and then discard it, this \emph{induces} an instrument $\Instr$ and a POVM $\Instr(1)$ on the target theory $\S$. Defining $\eta_\sigma : \S\otimes \mc{P} \to \S$ by linear and continuous extension of $A \otimes B \mapsto \sigma(B) A$ to ``trace out'' the probe degrees of freedom, the instrument $\Instr$ is given by linear maps $\Instr_\alpha :\S\to \S$,
\begin{equation}
\Instr_\alpha (\a) \defby \eta_{\sigma}(\Theta(\a\otimes M_\alpha )), \quad  \a\in \S,\quad \alpha \in \Out
\label{instr_FV}
\end{equation}
where it is easy to verify that these are completely positive and $\bar{\Instr}(1)=1$. Moreover, the induced POVM $\Instr (1)$ can be localized in any causally complete connected region containing $\Cou$.

A problem with the FV framework is that the set of allowed quantum operations is defined implicitly: Namely, any instrument of the form \eqref{instr_FV} is \emph{FV-realizable}. However, given an instrument $\Instr$ on the target theory, the problem of deciding whether $\Instr$ admits an FV representation is far from trivial: It amounts to determining if there exist a probe theory $\mc{P}$, a scattering morphism $\Theta$, a probe state $\sigma$ and a POVM $M$ such that \eqref{instr_FV} holds---this we call a \emph{measurement scheme} for $\Instr$. The problem of determining if a given POVM $M$ is realizable within the FV framework is similarly challenging. In \cite{FJR23}, Fewster, Jubb and Ruep prove that, for the case of a free scalar field, the set of FV-realizable instruments is rich enough to carry out state tomography: namely, it allows one to estimate all parameters defining the QFT state of the target theory. As illustrated at the end of the previous section, this is not the same as proving that any POVM can be realized within the FV framework. This begs the question: Can we characterize a class of relevant quantum instruments that are FV-realizable?

This question extends even further, if we regard the probe in the FV framework, not as a mere mathematical artifact to arrive at a well-behaved set of local QFT operations but as a (perhaps, effective) QFT which is at the root of any measurement we currently conduct in the lab. In that second case, the measurement of the probe should likewise be realizable within the FV framework. That would require introducing another probe, which in turn should be measured with another probe, and so on.

Now, let us call $\mathbbm{M}_1$ the set of POVMs that can be realized within the FV framework (i.e., of the form $\Instr(1)$ as given in \eqref{instr_FV} for some choice of a measurement scheme) and define the sets of POVMs $\{\mathbbm{M}_n\}_n$ by induction: A POVM is in $\mathbbm{M}_{n}$ if it admits an FV realization such that the probe is subject to a measurement in $\mathbbm{M}_{n-1}$. Analogously, we define $\mathbbm{I}_{n}$ as the set of instruments that admit an FV realization such that the probe is subject to a measurement in $\mathbbm{M}_{n-1}$. 

Following the recursion, we find that an instrument $\Instr$ is in $\mathbbm{I}_n$, resp., $I(1) \in \mathbbm{M}_n$ iff there is a sequence of measurement schemes $((\P_j,\Theta_j,\sigma_j,M^j))_{j=1}^n$ such that the $j^\text{th}$ scheme induces $M^{j-1}$ for $j > 1$ and the first scheme induces $\Instr$, resp., $\Instr(1)$. This we call a \emph{measurement chain}. By definition, the coupling regions $\Cou_j$ associated with $\Theta_j$ obey $\Cou_{j} \subset \Cou_{j-1}^+$ so that we may interpret the measurement chain also as a single measurement scheme $(\P,\Theta,\sigma,M)$ consisting of a ``super probe'' theory $\P = \P_1 \otimes \ldots \otimes \P_n$ in the probe state $\sigma$, the product state formed from $\sigma_1$,..,$\sigma_n$, a probe POVM $M = \id_{\P_1\otimes ..\otimes \P_{n-1}} \otimes M^n$ and a scattering morphism $\Theta = \hat{\Theta}_1 \circ \ldots \circ \hat{\Theta}_n$, where $\hat{\Theta}_j$ denotes the scattering morphism which acts as $\Theta_j$ on $\P_{j-1}\otimes \P_j$ and trivially on all other factors. The scheme $(\P,\Theta,\sigma,M)$ also induces $\Instr$, resp., $\Instr(1)$, and using $\iota_j \defby \eta_{\sigma_j} \circ \Theta_j$ we obtain\footnote{This identity is straightforward to show based on the factorization of the super probe state and its scattering morphism. It is analogous to the proof of \cite[Thm.~3.5]{FV20}.} the identity
\begin{align}\label{eq:chain}
    \Instr_\alpha (A) &= \eta_{\sigma}(\Theta(A \otimes M_\alpha)) \\
    &= \iota_1(A \otimes \iota_2(\id \otimes \ldots \iota_n(\id \otimes M^n_\alpha))), \quad A \in \S.
\end{align}

One wonders if there exist non-trivial POVMs $M$ or instruments $\Instr$---namely, with informative measurement outcomes---such that $M\in \mathbbm{M}_n$ or $\Instr \in \mathbbm{I}_n$ for all $n$. 

Since some quantum operations seem only reachable as limits of other FV-realizable quantum operations, it will be convenient to work with an asymptotic notion of the sets $\{\mathbbm{M}_n\}_n$ and $\{\mathbbm{I}_n\}_n$. A POVM $M$ on a QFT $\A$ belongs to $\overline{\mathbbm{M}}_n$ if there exists a sequence of POVMs $(M^k)_k \subset \mathbbm{M}_n$ such that $\lim_{k\to\infty} M_\alpha ^k =M_\alpha $ for all $\alpha \in \Out$. Analogously, an instrument $\Instr$ belongs to $\overline{\mathbbm{I}}_n$ if there exists a sequence of instruments $(\Instr^k)_k \subset \mathbbm{I}_n$ such that $\lim_{k\to\infty} \Instr^k_\alpha (\a) =\Instr_\alpha (\a)$ for all $\a \in \A$ and $\alpha \in \Out$.

In the next section, for linear scalar QFT, we will prove that Gaussian-modulated local field measurements and their associated Gaussian instruments are \emph{asymptotically FV-realizable}, i.e., they are elements of $\overline{\mathbbm{M}}_1$ and $\overline{\mathbbm{I}}_1$, respectively. Remarkably, the required probe measurement is also a Gaussian field measurement. Iterating our scheme, we then establish that Gaussian field measurements \emph{admit a movable FV-Heisenberg cut}, i.e., they are elements of $\overline{\mathbbm{M}}_n$ and $\overline{\mathbbm{I}}_n$, respectively, for all $n \in \mathbbm{N}$. Our precise results are as follows.

\vspace{1em}

\begin{thm}\label{thm:asymptFV}
    Suppose a QFT $\S$ induced by $T$ as described in \Cref{sec:alg} is given, and consider its Gaussian-modulated field measurements $M^{f,\epsilon}$ as given in \eqref{gaussian_POVM} and associated dephased Gaussian-modulated instruments $\Instr^{f,\epsilon,\delta}$ as given in \eqref{eq:dephasedgauss_instr}.
    Then, $M^{f,\epsilon} \in \overline{\mathbbm{M}}_n$ and $\Instr^{f,\epsilon,\delta} \in \overline{\mathbbm{I}}_n$ for all $n \in \N$, $f \in \SmoothC(\M)$ as well as arbitrary $\epsilon,\delta > 0$.
\end{thm}
\vspace{1em}
In words, all Gaussian-modulated measurements and all dephased Gaussian-modulated instruments of linear scalar fields are asymptotically FV-realizable and admit a movable FV-Heisenberg cut.

\newpage

\section{Implementation of Gaussian measurements\\ within the FV framework}\label{sec:impl}

In this section, we prove the results listed in \Cref{thm:asymptFV}. We consider two linear scalar QFTs, a target $\S$ and a probe $\mc{P}$, induced by, for simplicity, the same classical equation of motion $T\phi(x) = 0$, respectively, $T\psi(x)=0$, as described in \Cref{sec:quantumops}. We denote the associated quantized fields by $\Phi$ for the target and $\Psi$ for the probe. To implement the POVM $M^{f,\epsilon}$ and the instrument $\Instr^{f,\epsilon,\delta}$ in $\S$, corresponding to a Gaussian measurement of the target field mode $f$, we will use weak linear interactions and carefully chosen squeezed probe states. Our principal aim here is to define a measurement scheme $(\Theta_\lambda,\sigma_\lambda,M^\lambda)$ consisting of an interaction, a probe state, and a probe POVM, respectively, which induces $M^{f,\epsilon}$ and $\Instr^{f,\epsilon,\delta}$ in the limit $\lambda \to 0$.

\paragraph{Interaction:} 

We first use ideas from the tomographic asymptotic scheme presented in \cite{FJR23}, which solves the ``classical part'' of the scheme. There, it is shown that: For all $f \in \SmoothC(\Reg)$ localizable in a precompact region $\Reg \subset \M$ and all precompact regions $\Pro\subset \overline{\Reg}^+ = \M \setminus \J^-(\overline{\Reg})$ with $\Reg \subset D^-(\Pro)$, the past domain of dependence of $\Pro$, we find a region $\Cou \subset \Reg$, $g\in \SmoothC(\Pro)$ and $\beta \in \SmoothC(\Cou)$ such that
    \begin{equation}\label{eq:classscheme}
    f=-\beta E^-g.
    \end{equation} 
We fix $\Reg$ (localization region), $\Cou$ (coupling region), $\Pro$ (processing region) with a schematic setup given in \Cref{fig:scheme}, as well as $f$ (target field mode), $g$ (probe field mode) and $\beta$ as given above. For the following, we exclude the trivial case $[f] = 0$, which excludes also $[g]=0$: If we had $[g] = 0$ then $g = Tg'$ for some $g' \in \SmoothC(\M)$. However, by the previous assumptions, $f = -\beta E^- g$ as well as $\supp \beta \subset \Cou$ and $\supp g \subset \Pro$ so that $f = -\beta (E+E^+) T g' = 0$ since $E T \SmoothC(\M) = 0$ and $\supp \beta E^+ T g' \subset \supp \beta \cap \J^+(\supp Tg') = \Cou \cap \J^+(\Pro) = \emptyset$.

\begin{figure}[h]
    \centering
    \includegraphics[width=.7\textwidth]{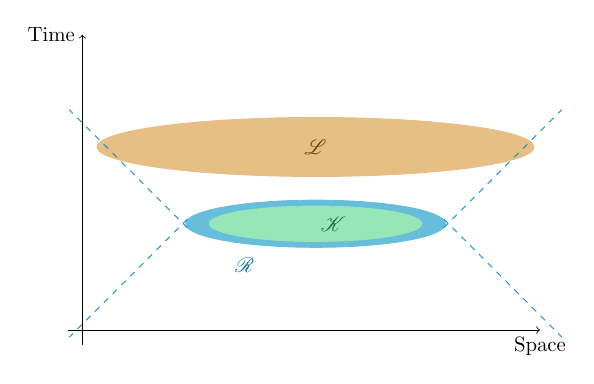}
    \caption{Schematic setup for the spacetime regions involved in the measurement scheme.}
    \label{fig:scheme}
\end{figure}

Following Fewster and Verch \cite[Sec.~4]{FV20}, we then define a linear and local interaction of strength $\lambda \beta$ between the probe and the target field where $\lambda$ is a real parameter. More precisely, we consider the interaction term
\begin{equation}
\lambda\int_\Cou \mathrm{vol}(x) \, \beta(x) \phi(x) \psi(x).
\end{equation}
This gives rise to a coupled equation of motion for the combined theory,
\begin{equation}
  T_\lambda \mat{\phi \\ \psi } = 0, \quad T_\lambda \defby \mat{T & \lambda \beta \\ \lambda\beta & T},
\end{equation}
which has well-defined Green operators for arbitrary real $\lambda$.

At the QFT level, this interaction generates a scattering morphism of the form
\begin{equation}
\Theta_\lambda\left(e^{i\Phi(u)} \otimes e^{i\Psi(v)}\right)= e^{i\Phi(u_\lambda)}\otimes e^{i\Psi(v_\lambda)}, \quad u,v \in \SmoothC(\M),
\label{def_theta_lambda}
\end{equation}
and, if we suppose that $u,v\in \SmoothC(\Cou^+)$, we have the explicit formula
\begin{equation}\label{eq:thetaexpl}
    \left(\begin{array}{c}u_\lambda\\v_\lambda\end{array}\right)=\theta_\lambda \left(\begin{array}{c}u\\v\end{array}\right), \quad \theta_\lambda \defby \left(\id-\lambda \left(\begin{array}{cc}0&\beta\\\beta&0\end{array}\right)E^-_{\lambda} \right),
\end{equation}
where $E^-_{\lambda}$ is the retarded Green operator associated with $T_\lambda$. For all such $u,v$, one finds from perturbation theory that
\begin{equation}
\left(\begin{array}{c}u_\lambda\\v_\lambda\end{array}\right)=\left(\begin{array}{c}u-\lambda\beta E^-v+O(\lambda^2)\\v-\lambda\beta E^-u+O(\lambda^2)\end{array}\right),
\label{pert_theory}
\end{equation}
where $E^-$ is the retarded Green operator associated with $T$ and the approximation holds in the standard test function topology.

For later use, considering arbitrary $h \in \SmoothC(\Pro)$, we define the functions $f_\lambda,g_\lambda, h_\lambda,p_\lambda$ through the identities:
\begin{align}
&\left(\begin{array}{c}\lambda f_\lambda\\ g_\lambda\end{array}\right)=\theta_\lambda\left(\begin{array}{c}0\\g\end{array}\right),\quad \left(\begin{array}{c} h_\lambda\\ \lambda p_\lambda\end{array}\right)=\theta_\lambda \left(\begin{array}{c}h\\0\end{array}\right).
\label{def_funcs}
\end{align}
By \eqref{pert_theory} we have that $\lim_{\lambda\to 0} g_\lambda=g$, $\lim_{\lambda\to 0} h_\lambda=h$ and $\lim_{\lambda\to 0} f_\lambda=-\beta E^- g = f$. Thus, by analogy we define
\begin{equation}
p:=\lim_{\lambda\to 0}p_\lambda=-\beta E^-h.
\label{def_p}
\end{equation}
We remark that
\begin{equation} \label{eq:symplcons}
    \braket{p,Eg} = \braket{f, Eh},
\end{equation}
 since
\begin{align}
&\langle p, E g\rangle=\langle p, (E^--E^+) g\rangle=\langle p, E^- g\rangle\nonumber\\
&=-\langle \beta E^-h, E^- g\rangle=-\langle E^-h, \beta E^- g\rangle=\langle E^-h, f\rangle=\langle Eh, f\rangle=\langle f, E h\rangle,
\end{align}
where we used in the first line that $\supp p \subset \Cou$ while $\supp E^+ g \subset \J^+(\supp g) = \J^+(\Pro)$ such that $\braket{p,E^+ g} = 0$ and in the second line that $\braket{f,E^+ h} = 0$ by an analogous argument.

\paragraph{Probe state:} Next, let us choose a family of initial states for the probe. They all arise starting from a fixed but arbitrary quasi-free state on $\P$ and applying a $\lambda$-dependent `squeezing' transformation. We define the involved transformations first abstractly:

Associated with a fixed mode pair $(q,\bar{q}) \in \mc{C}_T^2$ which is canonically conjugate, i.e., $E(q,\bar{q}) = 1$, we define the \emph{symplectic flip} $R_{q,\bar{q}}$ and, for nonzero $\mu \in \R$, the \emph{squeezing transformation} $F_{q,\bar{q},\mu}$ by
\begin{equation} \label{eq:symplflip}
    R_{q,\bar{q}} v \defby E(v,\bar{q})\bar{q} + E(v,q)q + v^{\perp},
\end{equation}
and
\begin{equation} \label{eq:symplsqueeze}
    F_{q,\bar{q},\mu} v \defby \mu E(v,\bar{q}) q - \mu^{-1} E(v,q) \bar{q} + v^{\perp},
\end{equation}
where 
\begin{equation}\label{eq:orth}
    v^{\perp} \defby v - E(v,\bar{q}) q + E(v,q) \bar{q}.
\end{equation}
These are evidently symplectic transformations characterized on the linear hull of $\{ q,\bar{q}\}$ by mapping $q \mapsto \bar{q}$, $\bar{q} \mapsto -q$, respectively, $q \mapsto \mu q$, $\bar{q} \mapsto \mu^{-1}\bar{q}$ and by acting trivially on the symplectic complement. Thus, given a state $\sigma$ and a canonically conjugate mode pair, we obtain a \emph{flipped state} $\tilde{\sigma}$ and a family $\sigma_\mu$ of $\mu$-\emph{squeezed states} by applying $R_{q,\bar{q}}$, respectively, $F_{q,\bar{q},\mu}$ to $\sigma$; confer \eqref{eq:transformedstate}. If $\sigma$ is quasi-free, then so will be $\tilde{\sigma},\sigma_\mu$, in this case their covariances $\tilde{\Gamma}$ and $\Gamma_\mu$ will satisfy
\begin{equation}
    \tilde{\Gamma}(q,q) = \Gamma(\bar{q},\bar{q}), \quad \tilde{\Gamma}(q,q) = \Gamma(\bar{q},\bar{q}), \quad \tilde{\Gamma}(q,\bar{q}) = - \Gamma(q,\bar{q})
\end{equation}
and
\begin{equation}
    \Gamma_\mu(q,q) = \mu^2\Gamma(q,q), \quad \Gamma_\mu(\bar{q},\bar{q}) = \mu^{-2}\Gamma(\bar{q},\bar{q}), \quad \Gamma_\mu(q,\bar{q}) = \Gamma(q,\bar{q}),
\end{equation} 
while 
\begin{equation}
    \tilde{\Gamma}(r,s)=\Gamma_\mu(r,s) = \Gamma(r,s), 
\end{equation} 
for all $r, s$ with 
\begin{equation}
E(r,q)=E(s,q)=E(r,\bar{q})=E(s,\bar{q})=0.    
\end{equation}

Now, let $g \in \SmoothC(\Pro)$ be as fixed earlier and $\bar{g} \in \SmoothC(\Pro)$ be a canonical conjugate to $g$\footnote{Note that $\Pro$ defines a globally hyperbolic spacetime in its own right implying that $\mc{C}_T(\Pro)$ has the properties listed for $\mc{C}_T = \mc{C}_T(\M)$ as listed in \Cref{sec:alg}. In particular, for any $g\in\SmoothC(\Pro)$ such that $[g] \neq 0$ within $\mc{C}_T$ and thus within $\mc{C}_T(\Pro)$, there always exists a canonically conjugate $\bar{g}\in \SmoothC(\Pro)$. Otherwise, we would have $\langle h, E g\rangle=0$, for all $h\in \SmoothC(\Pro)$ implying $Eg = 0$ and thus $g=Tg'$, for some $g'\in \SmoothC(\Pro)$, contradicting our initial assumption.}, so we have $E([g],[\bar{g}])=\braket{g,E\bar{g}}=1$. For our protocol, we will need a quasi-free state $\sigma$ with covariance $\Gamma$ satisfying the conditions:
\begin{equation}\label{tecreq}
\Gamma([g],[\bar{g}])=0, \quad \Gamma([g],[g])=\Gamma([\bar{g}],[\bar{g}])=:c^2,
\end{equation}
This is easily achieved: let $\sigma_0$ be an arbitrary quasi-free state on $\mc{P}$ with covariance $\Gamma_0$ and denote its flipped state (with respect to $[g]$ and $[\bar{g}]$) by $\tilde{\sigma}_0$ with covariance $\tilde{\Gamma}_0$. Then $\sigma$ can be defined as the quasi-free state with covariance $\Gamma \defby \frac{1}{2} (\Gamma_0 + \tilde{\Gamma}_0)$, which is well-defined and satisfies the requirement \eqref{tecreq} by inspection. In other words, we define $\sigma$ as the liberation\footnote{For each state there is a quasi-free state with identical covariance which is commonly referred to as the \emph{liberation} of the state; see e.g. \cite{Kay93}.} of the uniform mixture $\frac{1}{2}(\sigma_0+\tilde{\sigma}_0)$. 

Next, let $g_\lambda$ be as defined in \eqref{def_funcs}, then for each $\lambda \in [0,\lambda_0]$ for some $\lambda_0 > 0$, we choose a canonical conjugate $\bar{g}_\lambda \in \SmoothC(\Pro)$ to $g_\lambda$ such that 
\begin{equation} \label{eq:cconj}
    \underset{\lambda \to 0}{\lim} \bar{g}_\lambda = \bar{g}.
\end{equation}
This is always possible, e.g., by setting $\bar{g}_\lambda \defby (\braket{g_\lambda, E \bar{g}})^{-1} \bar{g}$, which is well-defined for sufficiently small $\lambda_0$ since $\braket{g_\lambda, E\bar{g}}$ is real and converges to 1 in the limit $\lambda \to 0$. Finally, we define $\sigma_\lambda$ as the $\mu\lambda-$squeezed state obtained from $\sigma$ with respect to the mode pair $([g_\lambda],[\bar{g}_\lambda])$,
\begin{equation} \label{def_sigma_lambda}
    \sigma_\lambda \defby \sigma_{F_{[g_\lambda],[\bar{g}_\lambda],\mu \lambda }},
\end{equation}
where we take $\lambda \in [0, \lambda_0]$ as a standing assumption and reserve a constant $\mu > 0$ for later optimization. 

Finally, note that if $\sigma$ is chosen to be Hadamard, then $\sigma_\lambda$ is Hadamard. This follows, since the Hadamard property is preserved under finite-rank perturbations (\Cref{prop:finhadamard}). Likewise, if $\sigma_0$ is chosen to be Hadamard, then also $\sigma$ is Hadamard, since, as can easily be verified, the Hadamard property is preserved under liberation and taking convex mixtures. Thus, our probe state preparation preserves the class of physically reasonable states: Starting with an arbitrary quasi-free Hadamard state, all probe states involved will have this property.

\paragraph{Probe measurement:} We have provided an interaction and a probe state. The last element we need to specify in the FV measurement scheme is the POVM to be implemented on the probe. We choose it to be Gaussian and dependent on $\lambda$. Namely, for some $\varepsilon > 0$,
\begin{equation}
M_\out ^{\lambda} \defby M_\out^{\lambda^{-1}g,\varepsilon} =\frac{1}{\sqrt{2\pi}\varepsilon} e^{-\frac{\left(\frac{\Psi(g)}{\lambda}-\out  \right)^2}{2\varepsilon^2}}.
\label{probe_POVM}
\end{equation}

\paragraph{Implementing the measurement scheme:}
We are ready to formulate our first result which establishes that Gaussian measurements admit an asymptotic FV-realization.
\begin{prop}\label{lemma_gaussian}
    Given a QFT $\S$ induced by $T$ as described in \Cref{sec:alg}, then the quadruple $(\S,\Theta_\lambda,\sigma_\lambda,M^\lambda)$, whose elements are respectively defined as in (\ref{def_theta_lambda}), (\ref{def_sigma_lambda}) and (\ref{probe_POVM}) for all $0 < \lambda \leq \lambda_0$ with implicit parameters $\varepsilon,\mu >0$ and $c$ as in \eqref{tecreq}, constitutes a measurement scheme which, in the limit $\lambda \to 0$, induces the POVM $M^{f,\epsilon}$ \eqref{gaussian_POVM} and the instrument $\Instr^{f,\epsilon,\delta}$ \eqref{eq:dephasedgauss_instr}, where
    \begin{equation}\label{epsdel}
    \epsilon=\sqrt{\varepsilon^2+c^2\mu^2},\quad  \delta=\sqrt{\frac{c^2}{\mu^2}-\frac{1}{4(\varepsilon^2+c^2 \mu^2)}}.
    \end{equation}
    Moreover, we can implement arbitrary (positive real) values of $\epsilon$ and $\delta$. It follows that
    \be M^{f,\epsilon} \in \overline{\mathbbm{M}}_1, \qquad \Instr^{f,\epsilon,\delta} \in \overline{\mathbbm{I}}_1. \non
    \ee
    for arbitrary $f \in \SmoothC(\M)$ and $\epsilon,\delta > 0$.
\end{prop}

\begin{proof}

To begin with, we prove that $\epsilon$ and $\delta$ as defined above can obtain arbitrary values. Note in this regard that while $\varepsilon$ and $\mu$ are free parameters of the protocol and thus can be taken arbitrarily small, $c^2$ in \eqref{tecreq} is constrained to be larger than $\frac{1}{2}$ by \eqref{eq:uncertaintyrel}. However, we can minimize the value of $c^2$ by exploiting our freedom in choosing a canonical conjugate for $q_\lambda = [g_\lambda]$ and in choosing our initial quasi-free probe state $\sigma$ with covariance $\Gamma$. In particular, if we choose our quasi-free state $\sigma'_0$ to be \emph{pure}, then, for any $\delta'>0$, we can find a conjugate $\bar{g}$ of $g$ such that $\Gamma'_0(g,g)\Gamma'_0(\bar{g},\bar{g})\leq (\frac{1}{2}+\delta')^2$. Squeezing $\sigma_0'$ appropriately, we obtain a state $\sigma_0$ with covariance matrix $\Gamma_0$ satisfying $\Gamma_0(g,g)=\Gamma_0(\bar{g},\bar{g})\leq\frac{1}{2}+\delta'$. Mixing it with its flipped version and liberating it, we obtain a state of the form (\ref{tecreq}), with $c^2\leq \frac{1}{2}+\delta'$. 

Thus, by variation of our free parameters $\varepsilon, \mu, c^2-\frac{1}{2} > 0$, we can obtain arbitrary positive values for $\epsilon$ and $\delta$ as defined in \eqref{epsdel}: To see this, let
\begin{equation}
    (\epsilon',\delta') = \underset{\varepsilon \to 0}{\lim} (\epsilon,\delta) = (\mu c, (\mu c)^{-1} \sqrt{c^4-\tfrac{1}{4}}).
\end{equation}
It is then easy to see that $(\epsilon',\delta')$ can take any value within $(0,\infty)^2$ by fixing $c = ((\epsilon' \delta')^2+\tfrac{1}{4})^{\tfrac{1}{4}}$ and $\mu = ((\epsilon' \delta')^2+\tfrac{1}{4})^{-\tfrac{1}{4}}\epsilon'$. Inserting this $\mu$ and $c$ instead in \eqref{epsdel}, where $\varepsilon$ is nonzero, for arbitrary but fixed $\epsilon',\delta'>0$, $(\epsilon,\delta)$ will approximate $(\epsilon',\delta')$ with arbitrary precision by choosing sufficiently small $\varepsilon > 0$.

Moving on, the instrument $\{\Instr^\lambda_\out \}_{\out  \in \R}$ induced by the measurement scheme $(\Theta_\lambda,\sigma_\lambda,M^\lambda)$ is defined through
\begin{equation}
\Instr^\lambda_\out (\a ):=\eta_{\sigma_\lambda}\left(\Theta_\lambda(\a\otimes M^\lambda_\out )\right), \quad \a \in \S.
\end{equation}
We will prove that
\begin{equation}
\lim_{\lambda\to 0}\Instr^\lambda_\out (\a )=\Instr^{f,\epsilon,\delta}_\out (\a ), \quad \a \in \S,
\end{equation}
which implies, in particular, $\lim_{\lambda \to 0} \Instr^\lambda_\out (1) = M_\out^{f,\epsilon}$, since $\Instr_\out ^{f,\epsilon,\delta}(1) = M_\out^{f,\epsilon}$ as discussed in \Cref{sec:quantumops}.

To characterize the instrument $\Instr^\lambda$, it suffices to study its action on Weyl operators $e^{i\Phi(h)}$ with arbitrary $h \in \SmoothC(\M)$, i.e., we compute
\begin{equation}
  \Instr^\lambda_\out \left(e^{i\Phi(h)}\right)=\frac{1}{\sqrt{2\pi}\varepsilon}\eta_{\sigma_\lambda}\left(\Theta_\lambda\left(e^{i\Phi(h)} \otimes e^{-\frac{\left(\frac{\Psi(g)}{\lambda}-\out \right)^2}{2\varepsilon^2}}\right)\right)
\end{equation}
According to \eqref{eq:lem1} in Appendix~\ref{app:gauss}, it holds that
\begin{equation}
  \frac{1}{\sqrt{2\pi}\varepsilon} e^{-\frac{\left(\frac{\Psi(g)}{\lambda}-\out \right)^2}{2\varepsilon^2}} = \int e^{-\frac{\varepsilon^2z^2}{2}}e^{-iz\out } e^{iz\frac{\Psi(g)}{\lambda}} \frac{dz}{2\pi}, \label{eq:fourgauss}
\end{equation}
where the integral here and below refers to the Bochner integral with respect to the strong-operator topology (as specified in the appendix). Thus, we first compute
\begin{equation}
  \eta_{\sigma_\lambda}\left(\Theta_\lambda\left(e^{i\Phi(h)} \otimes e^{iz\frac{\Psi(g)}{\lambda}} \right)\right).
\end{equation}

Since for any $h \in \SmoothC(\Cou^+)$ we find a processing region $\Pro$ (subject to the constraints from above) with $\supp h \subset \Pro$ and since $\mc{C}_T(\Cou^+) = \mc{C}_T$ (by the classical timeslice property), we assume $h \in \SmoothC(\Pro)$ without loss of generality. This allows us to evaluate the action of $\Theta_\lambda$ by \eqref{def_theta_lambda} and \eqref{eq:thetaexpl}. We obtain:

\begin{align}
&\eta_{\sigma_\lambda}\left(\Theta_\lambda\left(e^{i\Phi(h)} \otimes e^{iz\frac{\Psi(g)}{\lambda}}\right)\right)\nonumber\\
&=\eta_{\sigma_\lambda}\left(e^{i\Phi(h_\lambda+f_\lambda z)} \otimes e^{i\Psi(\frac{g_\lambda}{\lambda}z + \lambda p_\lambda)} \right)\nonumber\\
&=e^{i\Phi(h_\lambda+f_\lambda z)} \sigma_\lambda\left( e^{i\Psi(\frac{g_\lambda}{\lambda}z + \lambda p_\lambda)} \right)\nonumber\\
&=e^{i\Phi(h_\lambda)} e^{i(\Phi(f_\lambda) -\frac{1}{2}\langle f_\lambda, E h_\lambda\rangle)z}\sigma_\lambda\left(e^{i\Psi\left(\frac{g_\lambda}{\lambda}z+\lambda p_\lambda\right)} \right) \nonumber\\
&=e^{i\Phi(h_\lambda)} e^{i(\Phi(f_\lambda) -\frac{1}{2}\langle f_\lambda, E h_\lambda\rangle)z}\sigma\left(e^{i\Psi\left(\mu g_\lambda(z-\lambda^2\langle \bar{g}_\lambda, E p_\lambda\rangle) +\frac{\bar{g}_\lambda}{\mu}\langle g_\lambda,E p_\lambda\rangle+\lambda p_\lambda^\perp\right)}\right),\label{eq:step0}
\end{align}
where 
\be \label{eq:compl}
    p_\lambda^\perp \defby p_\lambda - E(p_\lambda, \bar{g}) g + E(p_\lambda,g)\bar{g};
\ee
similar\footnote{Note that the dependence on the choice of representatives, $g$ and $\bar{g}$, drops out since $e^{i\Psi(\cdot)}$ depends only on $\mc{C}_T$ not on $\SmoothC(\M)$.} to \eqref{eq:orth}.

Then we take \eqref{eq:fourgauss} and act on it with the map $\eta_{\sigma_\lambda} \left(\Theta_\lambda (e^{i\Phi(h)} \otimes \bullet )\right)$, which we can pull through the $z$-integral by linearity and continuity of the map (note here that $\eta_\sigma$ is linear and continuous by definition and that $\Theta_\lambda$ is an automorphism, also implying linearity and continuity). Inserting \eqref{eq:step0} we obtain
\begin{align}
  \Instr^\lambda_\out  & \left(e^{i\Phi(h)}\right) = e^{i\Phi(h_\lambda)} \nonumber \\
  & \times \int e^{-\frac{\varepsilon^2z^2}{2}}e^{i(\Phi(f_\lambda)-\out -\frac{1}{2}\langle f_\lambda, E h_\lambda\rangle)z}\sigma\left(e^{i\Psi\left(\mu g_\lambda(z-\lambda^2\langle \bar{g}_\lambda, E p_\lambda\rangle) +\frac{\bar{g}_\lambda}{\mu}\langle g_\lambda,E p_\lambda\rangle+\lambda p_\lambda^\perp\right)}\right) \frac{dz}{2\pi}.\label{eq:step1}
\end{align}

Next, let us take the limit $\lambda \to 0$ on \eqref{eq:step1}. For this note that the integrand in \eqref{eq:step1} is majorized by $e^{-\frac{\varepsilon^2 z^2}{2}}$ uniformly in $\lambda$. Thus, by dominated convergence for Bochner integrals \cite[Thm.~3]{DU77}, we can take the limit inside the integral and since $\sigma$ is quasi-free (thus regular, confer Appendix~\ref{app:aqft}) also inside $\sigma$. We obtain:
\begin{align}
&e^{i\Phi(h)}\int e^{-\frac{\varepsilon^2z^2}{2}}e^{i(\Phi(f)-\out -\frac{1}{2}\langle f, E h\rangle)z}\sigma\left(e^{i\Psi\left(\mu z g +\frac{\bar{g}}{\mu}\langle g,E p\rangle\right)}\right) \frac{dz}{2\pi} \nonumber\\
&=e^{i\Phi(h)}\int e^{-\frac{\varepsilon^2z^2}{2}}e^{i(\Phi(f)-\out -\frac{1}{2}\langle f, E h\rangle)z}e^{-\frac{c^2}{2}\left(\mu^2z^2+\frac{\langle g,E p\rangle^2}{\mu^2}\right)} \frac{dz}{2\pi}\nonumber\\
&=\frac{1}{\sqrt{2\pi(\varepsilon^2+c^2\mu^2)}}e^{i\Phi(h)}e^{-\frac{(\Phi(f)-\out -\frac{1}{2}\langle f, E h\rangle)^2}{2(\varepsilon^2+c^2\mu^2)}}e^{-\frac{c^2 \langle g, E p\rangle^2 }{2\mu^2}},
\label{eq_interm}
\end{align}
with $p$ defined as in (\ref{def_p}).

Moreover, since $\supp h \subset \Pro$, we have by \eqref{eq:symplcons} that $\langle p, E g\rangle =\langle f, E h\rangle$ and insertion into \eqref{eq_interm} yields
\begin{equation}
\lim_{\lambda\to 0}\Instr_\out ^\lambda\left(e^{i\Phi(h)}\right)=\frac{1}{\sqrt{2\pi(\varepsilon^2+c^2\mu^2)}} e^{-\frac{c^2\langle f, E h\rangle^2}{2\mu^2}} e^{i\Phi(h)}e^{-\frac{(\Phi(f)-\out -\frac{1}{2}\langle f, E h\rangle)^2}{2(\varepsilon^2+c^2\mu^2)}}.
\end{equation}
By comparison with \eqref{action_gauss_instr} we conclude the proof.
\end{proof}

\paragraph{Moving the ``FV Heisenberg cut'':}
After establishing that Gaussian measurements are asymptotically FV-realizable in \Cref{lemma_gaussian}, there are essentially two ways to extend our scheme; cf.~\cite[Sec.~4.3]{FV23}. One option is to treat multiple probes all coupled to the target theory, which allows for modeling successive Gaussian measurements on the target theory. The other is to treat a measurement chain as discussed in \Cref{sec:fv}. In this paragraph, we will show that for a suitable iteration of our scheme we obtain arbitrarily long measurement chains inducing Gaussian measurements asymptotically in the limit $\lambda \to 0$. Thus, we prove that Gaussian measurements admit a movable FV Heisenberg cut and complete the proof of our main theorem (\Cref{thm:asymptFV}). Our precise result is

\begin{prop}\label{prop_heisenberg}
    Let $\S$ be a QFT induced by $T$, as described in \Cref{sec:alg}. Then for every $n \in \mathbbm{N}$, the sequence $( (\P_j, \Theta^j_\lambda,\sigma^j_\lambda,M^{j,\lambda}))_{j=1}^n$, whose elements are respectively defined as in \eqref{eq:itint},\eqref{eq:itstate} and \eqref{eq:itpovm} for all $0 < \lambda \leq \lambda_0$ with implicit parameters $\varepsilon$, $\mu_1$,..,$\mu_n$ and $c_1,..,c_n$, defines a measurement chain with $n$ steps which, in the limit $\lambda \to 0$, induces the POVM $M^{f,\epsilon}$ \eqref{gaussian_POVM} and the instrument $\Instr^{f,\epsilon,\delta}$ \eqref{eq:dephasedgauss_instr}, where 
    \be \label{eq:epsdelpp}
        \epsilon = \sqrt{\varepsilon^2 + \sum_{j=1}^n \mu_j^2 c_j^2}, \quad \delta = \sqrt{ \frac{c_1^2}{\mu_1^2} - \frac{1}{4(\varepsilon^2 + \sum_{j=1}^n \mu_j^2 c_j^2)}}.
    \ee
    Moreover, since $\{\mu_j\}_{j=2}^n$ are arbitrary positive parameters, we can take $(\epsilon,\delta)$ as defined here as close as we want to the values of $(\epsilon,\delta)$ as defined in \eqref{epsdel} with $c=c_1$ and $\mu = \mu_1$. It follows that 
    \be \non
        M^{f,\epsilon}\in\overline{\mathbb{M}}_n, \qquad \Instr^{f,\epsilon,\delta}\in\overline{\mathbb{I}}_n
    \ee 
    for arbitrary $n \in \N$, $f \in \SmoothC(\M)$ and $\epsilon,\delta > 0$.
\end{prop}

\begin{proof}
The proof is easy, but cumbersome, so we omit some of the details. To begin with, we fix arbitrary $\varepsilon > 0$ and $n \in \N$ and suppose $j\in \{1,..,n\}$ throughout the proof. Also, we fix $f \in \mc{C}_T(\Reg)$ for some precompact $\Reg \subset \M$ and assume without loss of generality\footnote{If $[f]=0$, it implies trivially that $M^{f,\epsilon}$ and $\Instr^{f,\epsilon,\delta}$ are trivial and thus automatically contained in $\mathbbm{M}_n$, resp., $\mathbbm{I}_n$ for all $n \in \N$ and arbitrary $\epsilon,\delta > 0$.} that $[f] \neq 0$ within $\mc{C}_T$. We consider a target QFT $\S$ whose field we denote by $\Phi$ and $n$ probe QFTs $\{\mc{P}_j\}$ whose fields we denote by $\{ \Psi_j\}$, all induced by the same equation of motion operator $T$.

\begin{figure}[ht]
    \centering
    \includegraphics[width=.7\textwidth]{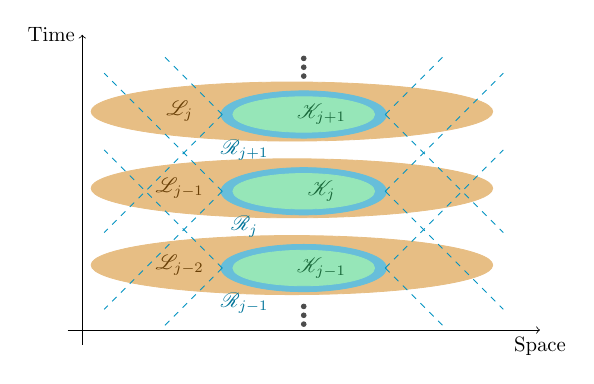}
    \caption{Schematic setup for the involved spacetime regions in the iterated scheme.}
    \label{fig:scheme_iterated}
\end{figure}

Next, we select localization regions $\{´\mathscr{R}_j\}$, coupling regions $\{ \Cou_j \}$ and processing regions $\{\Pro_j\}$ within $\M$ as well as probe modes $\{ g^j\}$ and interaction strengths $\{\beta_j\}$ within $\SmoothC(\M)$ suitable to induce the target mode $f$ classically. In particular, we choose $\{ \mathscr{R}_j\}$ and $\{ \Pro_j\}$ to be arbitrary precompact regions within $\overline{\Reg}^+$ such that $\Pro_j \subset \overline{\Pro_{j-1}}^+$ and $\mathscr{R}_j \subset \Pro_{j-1} \cap D^{-}(\Pro_{j})$ for all $j$; here $\Pro_0 = \Reg$. The involved spacetime regions are illustrated in \Cref{fig:scheme_iterated}. Then we invoke the classical part of the scheme from above iteratively: Starting with $j=1$ and $g^0=f$, we find regions $\Cou_j \subset \mathscr{R}_j$, $g^j \in \SmoothC(\Pro_j)$ and $\beta_j \in  \SmoothC(\Cou_j)$ such that 
\begin{equation}
    g^{j-1} = -\beta_{j} E^- g^{j};
\end{equation}
which implies also $[g^j] \neq 0$ for all $j$ since $[f] \neq 0$.

Given this, we can define the interaction. Fix $\lambda\in\R^+$. For each $j$, we consider the interaction term
\begin{equation}
\lambda\int_{\Cou_j} \mathrm{vol}(x) \, \beta_j(x)\psi_{j-1}(x)\psi_{j}(x)
\end{equation}
with $\psi_0 = \phi$. We define the corresponding scattering morphism $\Theta^j_{\lambda}$ on $\mc{P}_{j-1} \otimes \mc{P}_{j}$ with $\mc{P}_0 = \S$ in analogy with \eqref{def_theta_lambda}. That is, 
\begin{equation} \label{eq:itint}
\Theta^j_{\lambda}(e^{i\Psi_{j-1}(u^{j-1}_\lambda)} \otimes e^{i\Psi_{j}(u^j_\lambda)})=e^{i\Psi_{j-1}(u^{j-1})} \otimes e^{i\Psi_{j}(u^j)},    
\end{equation}
with    
\begin{equation}
\left(\begin{array}{c}u^{j-1}_\lambda\\ u^j_\lambda\end{array}\right) = \theta^j_{\lambda}\left(\begin{array}{c}u^{j-1}\\
u^j\end{array}\right)
\end{equation}
and $\theta^j_\lambda$ as in \eqref{eq:thetaexpl} for $\beta = \beta_j$.
    
Starting from $g^n_{-}(\lambda):=g^n$, we inductively define the $\lambda$-dependent functions $\{ g^j_{\pm}(\lambda) \}_j$ through the relation
\begin{equation}
\left(\begin{array}{c}\lambda g^{j-1}_{+}(\lambda) \\ g^j_{-}(\lambda)\end{array}\right)=\theta^j_{\lambda}\left(\begin{array}{c}0\\g^j_{-}(\lambda) \end{array}\right).
\label{def_g_pm}
\end{equation}
Note that $\lim_{\lambda\to 0} g^j_{\pm}(\lambda) =g^j$. We thus choose $\lambda$ small enough so that $[g^j_{\pm}(\lambda)]\not=0$ for all $j$. Hence, for every $g^j_{+}(\lambda)$, there exists a canonical conjugate $\bar{g}^j_{+}(\lambda)$.

Let $\sigma$ be a quasi-free probe state satisfying the technical condition \eqref{tecreq}, with $c = c_1$ and $g=g^1$. For each $j$ and arbitrary $\mu_j>0$, on $\P_j$, we define the probe state
\begin{equation}\label{eq:itstate}
\sigma^j_{\lambda}:=\sigma_{F_{[\bar{g}^j_{+}(\lambda)], [g^j_{+}(\lambda)],\mu_j\lambda^j}}.
\end{equation}

Starting from $M^n(\lambda)\defby M^{g^n_{-}(\lambda)\lambda^{-n},\varepsilon}$, we define the POVMs $(M^j(\lambda))_{j=1}^{n-1}$ through the recursive relation:
\begin{equation}
  M_\out^j(\lambda)=\eta_{\sigma^\lambda_j}\circ\Theta_\lambda^{j+1}(\id\otimes M^{j+1}_s(\lambda)).
\label{recur_POVMs}
\end{equation}
By construction, $M^j(\lambda)\in \mathbb{M}^{n+1-j}$, for $j=1,...,n$.
    
To arrive at an explicit expression for the $j^{th}$ POVM, we note that, for $\omega$ quasi-free, $g\in \SmoothC(\M)$ and $\varepsilon\in\R^+$, it holds that
\begin{equation}
\eta_{\omega}\circ\Theta_\lambda(\id\otimes M^{g,\varepsilon}_s)=M^{f_\lambda,\varepsilon'},
\label{explicit_recur_POVMs}
\end{equation}
with
\begin{equation}
\varepsilon'=\sqrt{\varepsilon^2+\lambda^{-2}\Gamma_\omega(g_\lambda,g_\lambda)},      
\end{equation}
where $g_\lambda, f_\lambda$ are defined from $g$ by the first identity in eq. (\ref{def_funcs}). From eqs. (\ref{def_g_pm}), (\ref{recur_POVMs}), (\ref{explicit_recur_POVMs}), it then follows that
\begin{equation}\label{eq:itpovm}
M^j(\lambda)=M^{g^{j}_{-}(\lambda)\lambda^{-j},\epsilon_j},
\end{equation}
with
\begin{equation}
\epsilon_j:=\sqrt{\varepsilon^2+\sum_{k=j+1}^n \mu_j^2 \Gamma(g_{+}^j(\lambda),g_{+}^j(\lambda))}.
\end{equation}

Next, we consider the effect of measurement $M^1(\lambda)$ on the target system. To do so, we note that, for $\omega$, quasi-free, $\varepsilon\in \R^+$, it holds that
\begin{equation}
\eta_{\omega}\circ \Theta_\lambda(e^{i\Phi(h)}\otimes M^{g,\varepsilon}_s)=e^{i\Phi(h_\lambda)}\int \frac{dz}{2\pi} e^{-\frac{\varepsilon^2z^2}{2}}e^{i(\phi(f_\lambda)-s-\frac{1}{2}\langle f_\lambda, E h_\lambda\rangle)z}\omega\left(e^{i\Psi\left(\frac{g_\lambda}{\lambda}z+\lambda p_\lambda^\perp \right)}\right),
\end{equation}
with $f_\lambda,g_\lambda,p_\lambda$ defined from $g,h$ through Eq. (\ref{def_funcs}) and $p^\perp_\lambda$ as in \eqref{eq:orth}, the symplectic complement of $p(\lambda)$ with respect to $\mbox{span}\{g^1(\lambda),\bar{g}^1(\lambda)\}$
    
Substituting $\omega$ by $\sigma_1$; $g$, by $\lambda^{-1} g^1_-(\lambda)$, and $\varepsilon$, by $\epsilon_1$, the instrument induced on the target system is:
\begin{equation}
\Instr_\out^\lambda=\int \frac{dz}{2\pi} e^{\frac{-\epsilon_1^2z^2}{2}}  e^{i(\Phi(g^0_{-}(\lambda))-s-\frac{1}{2}\langle g^0_{-}, E h_\lambda\rangle)z}\sigma\left(e^{i\Psi(\mu g^1_{+}(\lambda)(z-\lambda^2\langle \bar{g}^1_{+}(\lambda),E p(\lambda)\rangle +\lambda p^\perp(\lambda)}\right),
\label{final_instrument}
\end{equation}
where $h(\lambda),p(\lambda)$ are defined by the relation:
\begin{equation}
\left(\begin{array}{c}h(\lambda)\\ \lambda p(\lambda) \end{array}\right)=\theta^1_{\lambda}\left(\begin{array}{c}h\\
0\end{array}\right).
\end{equation}
Note that $\lim_{\lambda\to 0}h(\lambda)=h$, $\lim_{\lambda\to 0}p(\lambda)=-\beta_1E^-h$. 

By construction, $I^\lambda\in \mathbb{I}_n$. From Eq.~(\ref{final_instrument}), it is obvious that
\begin{equation}
\lim_{\lambda\to 0} I_s^\lambda=I^{f,\epsilon,\delta},
\end{equation}
with $\epsilon, \delta$ as defined in \eqref{eq:epsdelpp} and $c_j^2 \defby \Gamma(g^j,g^j)$.

\end{proof}

\paragraph{Projective measurements and recovery of the ``projection postulate'':}

The FV scheme outlined in the proof of \Cref{lemma_gaussian} allows one to realize many POVMs other than Gaussian-modulated field measurements. Indeed, consider a Gaussian measurement $\{ M_\out ^{f,\epsilon}\}_\out $ and let $\{p_{\tilde{\out }}:\R\to\R^+\}_{\tilde{\out }\in\tilde{\Out}}$ be any set of Borel-integrable functions with normalized parameter-dependence $\int_{\tilde{\Out}}p_{\tilde{\out }}(\out )d\tilde{\out }=1$, for all $\out \in\R$. Then the POVM $\{ \tilde{M}_{\tilde{\out }}^{f,\epsilon} \}_{\tilde{\out }}$ that results when we measure $\{ M_\out ^{f,\epsilon}\}_\out $ and, upon obtaining the result $\out $, sample $\tilde{\out }$ from the probability distribution $\{p_{\tilde{\out }}(\out )\}_{\tilde{\out }}$ will have the form
\begin{equation}
\tilde{M}_{\tilde{\out }}^{f,\epsilon}=\left(p_{\tilde{\out }} \ast g_\epsilon \right)\left(\Phi(f)\right), \quad g_\epsilon(x) \defby \tfrac{1}{\sqrt{2\pi}\epsilon}e^{-\frac{x^2}{2\epsilon^2}},
\label{meas_with_noise}
\end{equation}
using $\ast$ to denote convolution of functions. Since $\widehat{p_{\tilde{\out }} \ast g_\epsilon}(z) = e^{-\frac{\epsilon^2 z^2}{2}} \hat{p}_{\tilde{\out }}(z)$ is clearly integrable in $z$ for all $\epsilon > 0$, the proof of \Cref{lemma_gaussian} with initial probe measurement $\tilde{M}_\out^{\lambda^{-1}g,\epsilon}$ goes through and we obtain $\tilde{M}^{f,\epsilon} \in \overline{\mathbbm{M}}_1$. Taking the limit $\epsilon\to 0$, we have that $\tilde{M}_{\tilde{\out }}^{f,\epsilon}$ tends to $p_{\tilde{\out }}(\Phi(f))$ which is thus also within $\overline{\mathbbm{M}}_1$.

Combined with \Cref{prop_heisenberg}, this observation implies that any POVM of the form $\{ p_{\tilde{\out }}(\Phi(f))\}_{\tilde{\out }}$ also belongs to $\overline{\mathbbm{M}}_n$, for all $n\in \N$. Clearly, also discrete-outcome variants of these POVMs can be modeled; in this case normalization of $\{ p_j \}_j$ changes to $\sum_j p_j(\out ) = 1$. This includes the option $p_j=\chi_{B^j}$, where $\{B^j\subset \R\}_j$ is a discrete family of pairwise disjoint Borel-measurable sets that satisfy $\bigcup_j B^j=\R$. As a result, also the projective measurement $\{ \Pi_j^f\}_j$ with $\Pi^f_j \defby\Pi_{B^j}(\Phi(f))$ is asymptotically FV-realizable. 

However, note that, if we follow the construction of \Cref{lemma_gaussian}, then the measurement channel of the instrument associated to the POVM~\eqref{meas_with_noise} is always $\Instr^{f,\epsilon,\delta}$, regardless of how we choose to process the ``real" outcome $\out $. So, even in the limit $\epsilon\to 0$ and transforming to discretely many outcomes $j$, the measurement channel will not resemble anything like $\sum_{j} \Pi_j^f \bullet \Pi_j^f$ which would run again into Sorkin's causal paradox \cite{Sor93}. The FV formalism neatly avoids them, despite the fact that it allows measuring the projectors $\Pi^f_j$ up to arbitrary accuracy. The ``projection postulate" is restored only for observers in the causal complement region to the field mode $f$: For an observer who knows the classical outcome $\out$ to be contained in some Borel set $B \subset  \R$ but has access only to field modes $h$ which are supported in the causal complement to the support of $f$, we obtain Born's rule from \eqref{action_gauss_instr} in the (strong-operator) limit
\be
    \lim_{\epsilon \to 0} \, \int_B \Instr^{f,\epsilon,\delta}_\out (e^{i\Phi(h)}) d\out = \Pi_B(\Phi(f)) e^{i\Phi(h)} \Pi_B(\Phi(f)),
\ee
while the overall measurement channel is trivial,
\be
    \lim_{\epsilon \to 0} \, \overline{\Instr}^{f,\epsilon,\delta} (e^{i\Phi(h)}) = e^{i\Phi(h)}.
\ee 
We conclude that if the observer cannot resolve the causal structure of the apparatus which led to the measurement of $f$, then the projection postulate might be a valid idealization.

\paragraph{Movable ``Heisenberg cut'' in the purely *-algebraic framework:}
We conclude this section by noting that, even when $\A$ is just a *-algebra and not completed with respect to any topology (as done in the original formulation of the FV framework \cite{FV20}), it still allows to describe a large class of non-trivial, FV-realizable measurements with a movable FV-Heisenberg cut. Let $\A^f$ denote the *-algebra spanned by all abstract Weyl operators of the form\footnote{In order to keep the notation in line with the main text we write $e^{iz\Phi(f)}$ in place of the abstract symbols $W(zf)$ appearing in Appendix~\ref{app:aqft}. However, we don't suppose that the algebra elements or the symbols $\Phi(f)$ are operators on a Hilbert space or that they are continuous in $z$ or $f$.} $\{e^{iz\Phi(f)}:\, z\in\R\}$ and consider a finite-outcome POVM $\{M_j\}_j \subset \A^f$ such that, for some finite family $\{\a_j^k: \, j,k\}\subset \A^f$ and every outcome $j$,
\begin{equation}\label{eq:abspovm}
M_j=\sum_k \a_j^k(\a_j^k)^*.
\end{equation}
Given the (unambiguous) expression of $M_j$ as a finite linear combination of imaginary exponentials of $\Phi(f)$, define $\hat{M}_j(s)$ by replacing every instance of $\Phi(f)$ by $s$. Then, for any $\lambda\in\R$, $\{\hat{M}_j(\Psi(g_\lambda/\lambda))\}_j$ defines a POVM in the probe field. Moreover, the induced POVM element
\begin{equation}
\hat{M}^\lambda_j:= \eta_{\sigma_\lambda}\circ\Theta_\lambda\left(\hat{M}_j\left(\Psi\left(g_\lambda/\lambda\right)\right)\right)
\end{equation}
is also a positive semidefinite element of $\A$. One can then verify that
\begin{equation}
\lim_{\lambda\to 0}\omega(\hat{M}^\lambda_j)=\omega(M_j) \label{eq:explimit}
\end{equation}
for any regular state $\omega$ on $\A$, i.e., a state such that $\omega(e^{it\Phi(f)})$ is continuous in $t$ (see also Appendix~\ref{app:aqft} for details). Therefore, the FV framework---also in the purely algebraic setting---includes certain ``Weyl type" POVMs linearly generated from Weyl operators of a single field mode.

\section{Conclusion}\label{sec:conclusion}
We have argued that the tomographic results of \cite{FJR23} are insufficient to infer that the FV scheme is rich enough to model field measurements. Rather, it is proven there that averages of powers of the field operator can be estimated asymptotically. As explained in the text, this is not the same as being able to implement a (Gaussian-modulated) field measurement: in fact, their asymptotic tomographic scheme in the limit $\lambda\to 0$ induces the identity map as a measurement channel on the target QFT.

Instead, we have shown that Gaussian-modulated measurements of a (local) smeared field operator can be asymptotically implemented within the FV scheme. The degree $\epsilon$ of the modulation can be chosen arbitrarily low, which allows one to approximate a projective measurement of the field with arbitrary precision. For non-zero $\epsilon$, there is a simple map to update the QFT state conditioned on the measurement outcome: namely, the composition of the Gaussian instrument proposed by Jubb \cite{Jub22} and further studied by \cite{Oec24} with a single-mode dephasing channel.

Notably, we find that the probe measurement required to induce a Gaussian measurement in a QFT is itself a Gaussian measurement. That is, we can model the measurement carried out on this first probe by making it interact with a second probe, which is subsequently measured. This measurement, in turn, can be also conducted by making the second probe QFT interact with a third, and so on \emph{ad infinitum}.

It would be interesting to know if the property of having a movable FV-Heisenberg cut is applicable to all quantum measurements admitting an FV representation. The answer would be a resounding ``yes!" if all POVMs with elements localizable within a local algebra happened to be FV realizable.

\backmatter

\bmhead{Acknowledgements}

The authors thank Henning Bostelmann, Chris Fewster, Markus Fröb, Claudio Iuliano, Robert Oeckl, Maximilian Heinz Ruep, Leonardo Sangaletti, Rainer Verch for valuable discussions on the subject as well as remarks and questions which helped improving this text. 

\bmhead{Funding}

J.~Mandrysch was funded by the quantA core project ``Local operations on quantum fields".

\newpage

\begin{appendices}

\section{Abstract Weyl CCR-algebra}\label{app:aqft}

In this section, we motivate the Hilbert space setting of the main text from a purely algebraic construction of the Weyl CCR-algebra. The statements presented here are well known and can be gathered from text book accounts like \cite{Pet90} and \cite[Sec.~18]{HR15}. In this section, $(X,\spl)$ will be an arbitrary fixed symplectic space. The application to (linear scalar) QFT arises when $(X,\spl) = (\mc{C}_T,E)$ as defined in \Cref{sec:alg}. 

The \emph{*-Weyl algebra} $\mc{W}_{*}(X,\spl)$ over $(X,\spl)$ is defined as the unital *-algebra generated by symbols $\{ W(v), \, v \in X \}$ subject to the relations
    \begin{equation}
        W(v)^\ast = W(-v), \qquad W(v)W(w) = e^{-\frac{i}{2}\spl(v,w)} W(v+w).
    \end{equation}
The \emph{C*-Weyl algebra} $\mc{W}_{C*}(X,\spl)$ is obtained as the completion of $\mc{W}(X,\spl)$ with respect to its universal C*-norm. We denote $\mc{W}$ if our statements apply to both algebras.

It is well known (confer e.g. \cite[Sec.~18.1-2]{HR15}) that symplectic transformations of $X$ induce automorphisms on $\mc{W}$ via
\begin{equation}
    \alpha_F(W(u)) \defby W(Fu),\quad u\in X,
\end{equation}
where $F$ denotes the symplectic transformation and $\alpha_F$ the associated automorphism. Thus, given a state $\omega$ on $\mc{W}$ one may define a new state $\omega_F$ by 
\begin{equation}
   \omega_F(A) = \omega(\alpha_F(A)), \quad A \in \mc{W}.
\end{equation}

A state $\omega$ on $\mc{W}$ is called \emph{regular} iff the orbits $t \mapsto \omega(W(tv))$ are continuous for all $v \in X$. Note that all quasi-free states are regular since  $\omega(W(tv)) = e^{-t^2\Gamma(v,v)-itV(v)}$ is continuous in $t$ for all $v \in X$; here $V$ and $\Gamma$ denote the one- and two-point function of $\omega$.

Now, let $(X,\spl) = (\mc{C}_T,E)$ and let $(\pi_\omega,\H_\omega,\Omega_\omega)$ denote the GNS representation of $\mc{W}$ with respect to a regular state $\omega$. Then we obtain the setting of the main text. Namely, for each $f \in \SmoothC(\M)$ there is a unique self-adjoint operator $\Phi_\omega(f)$ on $\H_\omega$ such that
 \begin{equation}
    \pi_\omega(W([f])) = e^{i\Phi_\omega(f)}
 \end{equation}
 satisfying \eqref{eq:cont} and \eqref{eq:weyl} (cf. e.g. \cite[Sec.~18.3]{HR15}). Define $\A(\Reg) \defby \pi_\omega(\mc{W}(\Reg))'' \subset B(\H_\omega)$, where $'$ denotes the commutant algebra, then by the bicommutant theorem, $\A(\Reg)$ is thus a von-Neumann algebra. The resulting net of von-Neumann algebras $\{ \A(\Reg)\}_\Reg$ satisfies the usual axioms as given in the main text.

\newpage

\section{Preservation of Hadamard condition under finite-rank perturbations}\label{app:finrank}

In this section we show that the class of Hadamard states---and the subclass of quasi-free Hadamard states---of a real linear scalar field on a globally hyperbolic spacetime are preserved under finite-rank perturbations of the underlying classical phase space. While the result presented here might be known to experts, it has to the best of the author's knowledge not been addressed in the literature\footnote{Note though \cite[Thm.~C.4]{Ver94} which gives a similar argument using \cite[App.~B]{KW91} and which implies the result for the Klein-Gordon field under single-rank perturbations.}. However, this class of transformations, which we call \emph{finite-rank perturbations} for short, is very relevant in applications and arise for instance in the context of quantum optics whose instruments typically act only on a finite number of field modes (at least in idealization). The proof relies on the fact that such symplectic transformations induce smooth changes of the two-point function. Pending precise definitions which we will give below, we show that

\begin{thm}\label{prop:finhadamard}
    Let $\mc{W}(X,\spl)$ denote the *- or C*-Weyl algebra over the symplectic space $(X,\spl)$ as defined in Appendix~\ref{app:aqft}. Then the class of Hadamard states on Weyl-CCR algebra of a linear scalar field, $\mc{W}(\mc{C}_T,E)$, where $(\mc{C}_T,E)$ is defined as in \Cref{sec:alg}, is preserved under symplectic finite-rank perturbations.
\end{thm}

For our purposes giving a full definition of \emph{Hadamard states} on $\mc{W}(\mc{C}_T,E)$ is overly technical. Instead, we will state two facts on Hadamard states and otherwise refer to \cite[Sec.~2.4]{Hac16} for summarized results and references on Hadamard states and to \cite{San10} for a treatment including (generalized) linear scalar fields. Fact 1: Given a Hadamard state $\omega$ on $\mc{W}(\mc{C}_T,E)$ with covariance $\Gamma$ (i.e. symmetric part of the two-point function), then for any $q \in \mc{C}_T$, $\Gamma(\cdot,q)$, viewed as a distribution on $\M$, has a smooth kernel\footnote{Since $\Gamma$ is the symmetric part of the two-point function, it is sufficient to show that any Hadamard two-point function has this property. This is well known. Note that while most references like \cite[App.~B]{KW91} and \cite[Prop.~5.3.17(a)]{BDFY15} use a slightly less general setting ($T = \square_\M + W$ for some $W \in \Smooth(\M)$ or for $W = m^2 + \xi R$, respectively) than we have ($(T = \square_\M + V^\alpha \nabla_\alpha + W$ for some $V^\alpha,W \in \Smooth(\M)$), the wavefront set properties of $E$ and the Hadamard two-point function remain unchanged in our setting (compare e.g. \cite[Thm.~16]{BF09a}) and the argument of \cite[Prop.~5.3.17(a)]{BDFY15} goes through.}. Fact 2: Given two states $\omega$ and $\omega'$ on $\mc{W}(\mc{C}_T,E)$, where one of them is Hadamard, then the other is Hadamard iff the difference of their covariances, viewed as a bi-distribution on $\M$, has a smooth kernel\footnote{The reverse direction, which is everything we need here, holds more or less by definition: The addition of a bi-distribution with a smooth kernel does not change the singular part of the two-point function nor the wavefront set. The other direction, that the difference of two Hadamard two-point functions is smooth, can also be viewed as direct consequence of definition as given e.g. in \cite{KW91,Ver94} or proven from the wavefront set characterization as in \cite[Lemma~2.9+Prop.~3.2]{San10}.} Thus, in proving that the difference of the covariance of the transformed and untransformed state are smooth in \Cref{lem:smoothdiff} below, we will validate \Cref{prop:finhadamard}.

\begin{rem}
    It is well known that the class of quasi-free states is preserved under \emph{all} symplectic transformations (confer e.g. \cite[Prop.~29.3-2]{HR15}). Thus, as a corollary, we also show that the class of quasi-free Hadamard states is preserved under symplectic finite-rank perturbations.
\end{rem}

To prepare the proof, we begin with defining some elementary notions of symplectic linear algebra. Given a symplectic space $(X,\spl)$, where $\spl$ denotes the symplectic form, a linear transformation $F$ of $X$ is referred to as a \emph{symplectic transformation} (of $X$) iff $\spl(Fu,Fv) = \spl(u,v)$ for all $u,v \in X$. Given a linear subspace $S \subset X$, its symplectic complement $S^\perp \defby \{ u \in X : \spl(u,S) = 0\}$ is another linear subspace of $X$ which is orthogonal to $S$, i.e., $\spl(S,S^\perp) = 0$. A linear subspace $S \subset X$ is referred to as \emph{symplectic} iff $S \cap S^\perp = \{ 0 \}$ and in this case the restriction of $\spl$ to $S$ defines a symplectic form on $S$. Further, if $S$ has dimension $2n$ for some $n \in \mathbbm{N}$ (any finite-dimensional symplectic space has even dimension), we say that $u$ defines a \emph{standard basis} of $S$ iff $u=(r_1,s_1,\ldots,r_n,s_n) \in S^{2n}$ such that $\spl(r_{j},r_k) =\spl (s_j,s_k) = 0$ and $\spl(r_j,s_k) = \delta_{jk}$ for all $j,k=1,..,n$ and it is helpful to introduce the \emph{symplectic matrix} $J$ defined through $J_{lm} \defby \spl(u_l,u_m)$, $l,m=1,..,2n$, which is independent from the choice of $u$ and summarizes the defining relations of a standard basis.

\begin{defn}
    A symplectic transformation $F$ is referred to as a (symplectic) \emph{finite-rank perturbation} iff $F-1$ has finite-rank.
\end{defn} 
Note that $F-1$ having finite-rank means, by definition, that the range of $F-1$ is finite-dimensional. This is equivalent to the existence of a finite-dimensional symplectic subspace $S \subset X$ such that $F\restriction_S$ is a symplectic transformation of $S$ and $F\restriction_{S^\perp} = 1$. It is even sufficient to check that $F\restriction_{S^\perp} = 1$ for some finite-dimensional symplectic subspace $S \subset X$: Let $P_S$ denote the orthogonal projection from $X$ onto $S$. Then $F\restriction_S$ maps into $S$ since $\spl((1-P_S)FP_Sv,w) = \spl(FP_Sv,(1-P_S)w) = \spl(FP_Sv,F(1-P_S)w) = \spl(P_Sv,(1-P_S)w) = 0$ for all $v,w \in X$ which implies $(1-P_S)FP_S = 0$ by nondegeneracy of $\spl$. Thus $\operatorname{ran} F\restriction_S \subset S$. Moreover, $P_S F(1-P_S)=0$ by assumption. That $F\restriction_S$ is a symplectic transformation of $S$ is then immediate since $F$ is a symplectic transformation of $X$. Finally, for a symplectic subspace $S \subset X$, we refer to a symplectic transformation $F$ of $X$ as the \emph{canonical extension} of $F\restriction_S$ to $X$ iff $F\restriction_{S^\perp} = \id$. Given $F\restriction_S$, the canonical extension to $X$ is uniquely given\footnote{Clearly, the given expression is a well-defined canonical extension. Uniqueness follows since $FP_S=F\restriction_S$ and $F(1-P_S) = 1$ are required by definition and imply $P_SFP_S = F\restriction_S$, $(1-P_S)F(1-P_S)\restriction_{S^\perp}=1$ as well as $(1-P_S)FP_S = P_SF(1-P_S) = 0$.} by $F = F\restriction_S P_S + (1-P_S) = (F\restriction_S-1)P_S + 1$. A helpful lemma for us is

\begin{lemma}\label{lem:finiterankproj}
    Given a symplectic space $(X,\spl)$, then a symplectic finite-rank perturbation $F$ of $X$ takes the form
    \begin{equation} \label{eq:finrank}
        Fq = \sum_{j,k=1}^{2n} J_{jk} \spl(q,u_k) (F-1)u_j + q, \quad q \in X,
    \end{equation}
    where $u$ is any standard basis of some $2n$-dimensional symplectic subspace $S \subset X$ such that $F \restriction_{S^\perp} = \id$ and $J$ is the symplectic matrix. 
\end{lemma}

\begin{proof}
    Since $F$ is finite-rank, there exists a $2n$-dimensional symplectic subspace $S \subset X$ such that $F\restriction_{S^\perp} = 1$ for some $n \in \mathbbm{N}$. Then, given a standard basis $u$ on $S$, it is straightforward to check that $P_S$ takes the explicit form
    \begin{equation}
       P_S=\sum_{j,k=1}^{2 n} J_{jk} \spl(\cdot, u_k) u_j
    \end{equation}
    and that $F = (F\restriction_S-1) P_S + 1$ takes the form as in \eqref{eq:finrank}. It is clear that this holds independent of the choice of $S$ and $u$ as long as $S$ is symplectic and finite-dimensional as well as $F\restriction_{S^\perp} = 1$.
\end{proof}

Elementary examples for finite-rank perturbations are given by the symplectic flip \eqref{eq:symplflip} and the (single-mode) squeezing transformation \eqref{eq:symplsqueeze} as given in the main text.

\begin{lemma}\label{lem:smoothdiff}
    Given a Hadamard state $\omega$ on $\mc{W}(\mc{C}_T,E)$ with covariance $\Gamma$ and a symplectic finite-rank perturbation $F$ of $\mc{C}_T$, then the bilinear form $\Delta(f,g) \defby \Gamma_F([f],[g]) - \Gamma([f],[g])$, $f,g \in \SmoothC(\M)$ has a smooth kernel.
\end{lemma}

\begin{proof}
    Representing $F$ as in \Cref{lem:finiterankproj} for suitable $u$ and $S$, we compute
    \begin{align}
        & \Delta(q,r) \non \\
        & \quad =\Gamma_F(q,r)-\Gamma(q,r) = \Gamma(Fq,Fr) - \Gamma(q,r) \\
        & \quad = \Gamma \left( \sum_{j,k=1}^{2n} J_{jk} E(q,u_k) (F-1)u_j + q , \sum_{lm=1}^{2n} J_{lm} E(r,u_{m}) (F-1)u_{l} + r \right) - \Gamma(q,r) \\
        & \quad = \sum_{j,k,l,m=1}^{2n} J_{jk} J_{lm} E(q,u_k) E(r,u_{m}) \Gamma( (F-1)u_j,   (F-1)u_{l}) \non \\
        & \qquad + \sum_{j,k}^{2n} J_{jk} E(q,u_k) \Gamma((F-1)u_j,r) +  \sum_{l,m}^{2n} J_{lm} E(r,u_{m}) \Gamma(q,(F-1)u_{l}).
    \end{align}
    Thus, $\Delta$ can be expressed as a finite linear combination of products of $E(\cdot,u)$ and $\Gamma(\cdot,u)$ for suitable $u \in \mc{C}_T$. Both, viewed as distributions on $\M$, have smooth kernels: For the first note that $Ef \in \Smooth(\M)$ for any $f\in \SmoothC(\M)$; see e.g. \cite[Sec.~3.4]{BGP07}. For the second, this follows since $\Gamma$ is the covariance of a Hadamard state; confer the beginning of this section.
\end{proof}

\section{Proofs on Gaussian instruments}\label{app:gauss}

In this section, we prove the most relevant statements made in \Cref{sec:quantumops}. Note also the complementary exposition in \cite[Secs.~3,4]{Oec24} which admits more mathematical details on Gaussian instruments but provides only a weak-operator limit to projective measurements and does not cover dephasing.

Concerning setting and notation, our results apply to an arbitrary (possibly nonseparable) Hilbert space $\H$, the vector norm on $\H$ will be denoted by $\lVert \cdot \rVert_\H$ and on $B(\H)$ we employ the seminorms $\lVert \cdot \rVert_\varphi \defby \lVert \cdot \varphi \rVert_\H$ as well as the (operator) norm $\lVert \cdot \rVert$, generating respectively the strong operator and the norm topology on $B(\H)$. Further, we denote by $\A(f) \subset B(\H)$ the von-Neumann algebra generated by $e^{i\Phi(f)}$ for a fixed $f \in \SmoothC(\M)$. In agreement with the main text a continuous function $g : \R \to \A(f)$ is termed \emph{(strongly) integrable} iff $r \mapsto \lVert g(r) \rVert_\varphi$ is an integrable function for all $\varphi\in \H$. This implies that $\int g(r) dr \in \A(f)$ defined via $\int g(r)dr \, \varphi \defby \int g(r)\varphi dr$ for all $\varphi\in \H$ exists as a Bochner integral on $\H$ (see e.g. \cite[Sec.~II.2]{DU77} for a text book account on Bochner integration) and aligns with the terminology `Bochner integral in the strong-operator topology' used in the main text. Since $\lVert a\rVert_\varphi \leq \lVert a \rVert$ for all $\varphi\in \H$ with $\lVert \varphi \rVert_\H = 1$, strong integrability follows also if $\lVert g(\cdot) \rVert$ is integrable.

\begin{lemma}[Gaussian measurement]\label{lem:gaussPOVM}
    For any $\epsilon > 0$ and $f \in \SmoothC(\M)$, $\{ M_\out ^{f,\epsilon} \}_{\out \in \R}$ as defined in \eqref{gaussian_POVM} yields a continuous outcome POVM. Moreover, we find the strong-operator limit $\epsilon \to 0$ as given in \eqref{eq:gaussPOVMlimit} and the useful formula \eqref{eq:lem1} holds.
\end{lemma}

\begin{proof}
    We fix arbitrary $\epsilon > 0$ and $g_\epsilon(r) \defby \frac{1}{\sqrt{2\pi} \epsilon} e^{-\frac{r^2}{2\epsilon^2}}$ for any $r \in \R$. Given that $r \mapsto g_\epsilon(r-\out ) $ is a nonnegative bounded smooth function on $\R$, by functional calculus $M_\out ^{f,\epsilon} = g_\epsilon(\Phi(f)-\out )$ defines a positive element of $\A(f)$. Now, for any bounded Borel set $B \subset \R$, also $h_{B,\epsilon} (r) \defby \int_B g_\epsilon(r-\out )d\out $ defines a nonnegative bounded smooth function on $\R$ with well-defined pointwise limits $\epsilon \to 0$ and $\Lambda \to \infty$ for $B = [-\Lambda,\Lambda]$, $\Lambda > 0$:
    \begin{equation}
        \underset{\epsilon \to 0}{\lim} \, h_{B,\epsilon}(r) = \chi_B(r), \quad \underset{\Lambda \to \infty}{\lim} \, h_{[-\Lambda,\Lambda],\epsilon}(r) = 1
    \end{equation}
    Since $|h_{B,\epsilon}(r)|\leq 1$ for all $r\in \R$ and $h_{B,\epsilon}(\Phi(f)) = \int_B M_\out ^{f,\epsilon} d\out  \in \A(f)$ this implies, referring for instance to \cite[Prop.~4.12\pl v\pr]{Sch12}, the strong-operator limits
    \begin{equation}
        \underset{\epsilon \to 0}{\lim} \, \int_B M_\out ^{f,\epsilon}d\out  = \Pi_B(\Phi(f)), \quad \int M_\out ^{f,\epsilon}d\out  = \underset{\Lambda \to \infty}{\lim} \, \int_{-\Lambda}^{\Lambda} M_\out ^{f,\epsilon}d\out  = 1.
    \end{equation}
    
    Finally, since $g_\epsilon \in L^1(\R)$ its Fourier transform is well-defined and yields $\hat{g}_\epsilon(z) = \int g_\epsilon(r) e^{-izr} dr = e^{-\frac{\epsilon^2z^2}{2}}$. Thus, we obtain
    \begin{equation} \label{eq:lem1}
      M_\out ^{f,\epsilon} = \int e^{-\frac{\epsilon^2 z^2}{2}} e^{-iz(\Phi(f)-\out )} \frac{dz}{2\pi};
    \end{equation}
    where strong integrability follows since the integrand is (strongly operator) continuous wrt $z$ and majorized (in norm) by $z \mapsto e^{-\frac{\epsilon^2 z^2}{2}} \in L^1(\R)$.
\end{proof}

\begin{lemma}[Gaussian and dephasing instrument]\label{lem:instruments}
    For any $\epsilon,\delta > 0$ and $f \in \SmoothC(\M)$, $\{ \Instr_\out ^{f,\epsilon}(\bullet) \}_{\out \in\R}$ as defined in \eqref{gaussian_instrument_basic} yields a continuous outcome instrument and $D^{f,\delta}(\bullet)$ as defined in \eqref{eq:dephasing} yields a (single outcome) instrument.
\end{lemma}

\begin{proof}
    In this proof, let $f \in \SmoothC(\M)$ be fixed and drop all such indices. By definition $\Instr_\out ^{\epsilon}(\bullet) = \sqrt{M_\out ^{\epsilon}} \bullet \sqrt{M_\out ^{\epsilon}}$ and $D^{\delta}(\bullet) = \int N_\nu^{\delta} \bullet (N_\nu^{\delta})^\dagger d\nu$, for $N_\nu^{\delta} \defby (2\pi)^{-\frac{1}{4}} \delta^{-\frac{1}{2}} e^{-\frac{\nu^2}{4\delta^2}} e^{-i\nu\Phi(f)}$; decomposing them into their Kraus operators $\{ \sqrt{M_\out ^{\epsilon}} \}_\out $ and $\{ N_\nu^{\delta}\}_\nu$ implying complete positivity. 
    
    The $\nu$-integral is well-defined as a strong integral since $N_\nu^{\delta}$ is (strongly operator) continuous in $\nu$ and $\lVert N_\nu^{\delta} \a (N_\nu^{\delta})^\dagger \rVert \leq (2\pi)^{-\frac{1}{2}} \delta^{-1} e^{-\frac{\nu^2}{2\delta^2}} \lVert \a \rVert$ is integrable in $\nu$ for any $\a \in B(\H)$. Pointwise strong integrability of $\Instr_\out ^{\epsilon}$ follows by strong integrability of $M_\out ^{\epsilon}$ as proven in the preceding lemma and simple operator estimates: For any $\varphi \in \H$ with $\lVert \varphi \rVert_\H = 1$ and any $\a \in B(\H)$, based on $\lVert \sqrt{M_\out ^{\epsilon}} \rVert \leq (2\pi)^{-\frac{1}{4}} \sqrt{\epsilon}^{-1}$ and since $\sqrt{M_\out ^{\epsilon}} = 2^{\frac{3}{4}} \pi^{\frac{1}{4}} \sqrt{\epsilon} M_\out^{\sqrt{2}\epsilon}$, it follows that
    $$\lVert \sqrt{M_\out ^{\epsilon}} \a \sqrt{M_\out ^{\epsilon}} \rVert_{\varphi} \leq \sqrt{2} \lVert \a \rVert \lVert M_\out ^{\sqrt{2}\epsilon} \rVert_\varphi ;$$
    with the r.h.s. being integrable.
    
    That units are preserved follows based on $\int M_\out ^{\epsilon} d\out  = 1$ and 
    $$\int N_\nu^{\delta} (N_\nu^{\delta})^\dagger d\nu = \frac{1}{\sqrt{2\pi}\delta} \int e^{-\frac{\nu^2}{2\delta^2}} d\nu = 1.$$
\end{proof}

\begin{lemma}[Commutation relation]\label{lem:comrel}
    Let $\mu = \hat{l}$ be the Fourier transform of a function $l \in L^1(\R)$. Then $\mu(\Phi(f)) \in \A(f)$ exists by functional calculus and it holds that
    \begin{equation}\label{eq:comrel}
      \mu(\Phi(f))e^{i\Phi(h)}=e^{i\Phi(h)}\mu(\Phi(f)-\langle f,E h\rangle).
  \end{equation}
\end{lemma}

\begin{proof}
    By Riemann-Lebesgue lemma $\mu$ is a bounded continuous function and thus $\mu(\Phi(f)) = \int l(z) e^{-iz\Phi(f)} \frac{dz}{2\pi} \in \A(f)$ exists by Borel functional calculus as a strong integral on $\A(f)$ based on (strong operator) continuity of $e^{-iz\Phi(f)}$ in $z$ and integrability of $l$. Then, we compute (using the Weyl relation),
    \begin{align}
      \mu(\Phi(f)) e^{i\Phi(h)} & = \int l(z) e^{-iz\Phi(f)} e^{i\Phi(h)} dz \non \\
      &= e^{i\Phi(h)} \int l(z) e^{iz(\langle f, E h \rangle -\Phi(f))} = e^{i\Phi(h)} \mu(\Phi(f)-\langle f, E h \rangle \id).
    \end{align}
\end{proof}

\begin{lemma}[Dephased gaussian instrument]\label{lem:dephasedgauss}
    For any $\epsilon, \delta > 0$, $\{ \Instr^{f,\epsilon,\delta}_\out \}_{\out \in \R}$ as defined in \eqref{eq:dephasedgauss_instr} defines a continuous outcome POVM and it holds for arbitrary $h \in \mc{C}_T$ that
    \begin{equation}
      \Instr^{f,\epsilon,\delta}_\out (e^{i\Phi(h)})=\frac{1}{\sqrt{2\pi}\epsilon}e^{-\frac{\langle f,Eh\rangle^2}{8\underline{\epsilon}(\epsilon, \delta)^2}} e^{i\Phi(h)} e^{-\frac{(\Phi(f)-\out -\frac{\langle f,Eh \rangle}{2})^2}{2\epsilon^2}}, \quad \underline{\epsilon}(\epsilon, \delta)^2:= \frac{1}{4\delta^2+\frac{1}{\epsilon^2}}.
    \end{equation}
\end{lemma}

\begin{proof}
   Since $\Instr^{f,\epsilon,\delta}_\out  = D^{f,\delta} \circ \Instr^{f,\epsilon}_\out  = \Instr^{f,\epsilon}_\out  \circ D^{f,\delta}$ is defined as a concatenation of two instruments (\Cref{lem:instruments}), it is itself an instrument (there are no questions arising concerning integrability since $D^{f,\delta}$ is single outcome). Invoking the identity \eqref{eq:comrel} for $\mu$ being a Gaussian function, completing squares and performing Gaussian integrals we find:
\begin{align}
  &\Instr^{f,\epsilon,\delta}_\out (e^{i\Phi(h)})\nonumber \\&=\frac{1}{2\pi\delta\epsilon}\int e^{-\frac{\nu^2}{2\delta^2}} e^{-\frac{(\Phi(f)-\out )^2}{4\epsilon^2}-i\nu\Phi(f)} e^{i\Phi(h)} e^{-\frac{(\Phi(f)-\out )^2}{4\epsilon^2}+i\nu\Phi(f)} d\nu \nonumber\\
  &=e^{i\Phi(h)}\frac{1}{2\pi\delta\epsilon}\int e^{-\frac{\nu^2}{2\delta^2}} e^{-\frac{(\Phi(f)-\out -\langle f,Eh\rangle)^2}{4\epsilon^2}-i\nu(\Phi(f)-\langle f,Eh\rangle)} e^{-\frac{(\Phi(f)-\out )^2}{4\epsilon^2}+i\nu\Phi(f)} d\nu \nonumber\\
  &=e^{i\Phi(h)}\left(\frac{1}{\sqrt{2\pi}\delta}\int e^{-\frac{\nu^2}{2\delta^2}} e^{i \nu \langle f,E h\rangle } d\nu \right)\left(\frac{1}{\sqrt{2\pi}\epsilon}e^{-\frac{(\Phi(f)-\out -\frac{\langle f,Eh \rangle}{2})^2}{2\epsilon^2}} e^{-\frac{\langle f,Eh \rangle^2}{8\epsilon^2}}\right)\nonumber\\
  &=\frac{1}{\sqrt{2\pi}\epsilon}e^{i\Phi(h)}e^{-\frac{\langle f,Eh\rangle^2\delta^2}{2}}e^{-\frac{(\Phi(f)-\out -\frac{\langle f,Eh \rangle}{2})^2}{2\epsilon^2}} e^{-\frac{\langle f,Eh \rangle^2}{8\epsilon^2}}\nonumber\\
  &=\frac{1}{\sqrt{2\pi}\epsilon}e^{-\frac{\langle f,Eh\rangle^2}{8\underline{\epsilon}(\epsilon, \delta)^2}} e^{i\Phi(h)}e^{-\frac{(\Phi(f)-\out -\frac{\langle f,Eh \rangle}{2})^2}{2\epsilon^2}}.
\end{align}
\end{proof}

\end{appendices}

\newpage

\bibliography{sn-bibliography}

\end{document}